\documentclass[accepted]{article}

\usepackage{microtype}
\usepackage{graphicx}
\usepackage{subfigure}
\usepackage{booktabs} 

\usepackage{hyperref}
\usepackage[most]{tcolorbox}

\usepackage[accepted]{icml2024}

\usepackage{amsmath}
\usepackage{amsfonts}
\usepackage{mathrsfs}
\usepackage{amssymb}
\usepackage{amsthm}
\usepackage{mathabx}
\usepackage{array,multirow,makecell}
\usepackage[mathcal]{euscript}
\usepackage[normalem]{ulem}
\usepackage{stmaryrd}
\usepackage{bbold}
\usepackage{paralist}

\usepackage[capitalize,noabbrev]{cleveref}

\theoremstyle{plain}
\newtheorem{theorem}{Theorem}[section]

\newtheorem{lemma}[theorem]{Lemma}

\theoremstyle{definition}
\newtheorem{definition}[theorem]{Definition}

\theoremstyle{remark}

\usepackage[textsize=tiny]{todonotes}

\usepackage{pgfplots}
\usepackage{tikz}
\usetikzlibrary{shapes}
\usetikzlibrary{snakes}
\usetikzlibrary{arrows}
\usetikzlibrary{automata}
\usetikzlibrary{fit}
\usetikzlibrary{shadows}
\usetikzlibrary{trees}
\usetikzlibrary{decorations}
\usetikzlibrary{decorations.text}
\usetikzlibrary{positioning}
\usetikzlibrary{patterns}
\usetikzlibrary{backgrounds}
\usetikzlibrary{matrix}
\usetikzlibrary{arrows.meta}
\usetikzlibrary{calc, shapes}
\usetikzlibrary {petri} 
\usetikzlibrary{chains,scopes}

\usepackage{xspace}
\def\ie{{\em i.e.,}\xspace}
\def\eg{{\em e.g.,}\xspace}
\def\cf{{\em cf.}\xspace}
\def\wrt{{\em w.r.t.}\xspace}

\DeclareMathOperator*{\argmax}{\arg\!\max}
\DeclareMathOperator*{\argmin}{\arg\!\min}

	\definecolor{sthlmLightBlue}{RGB}{214,237,252} 
		
	\definecolor{sthlmBlue}{RGB}{0,110,191} 

	\definecolor{sthlmLightGreen}{RGB}{213,247,244} 
		
	\definecolor{sthlmGreen}{RGB}{0,134,127} 

	\definecolor{sthlmLightGrey}{RGB}{213,217,225} 
		
	\definecolor{sthlmGrey}{RGB}{245,243,238} 
		
	\definecolor{sthlmDarkGrey}{RGB}{51,51,51} 

	\definecolor{sthlmLightOrange}{RGB}{255,215,210} 
		
	\definecolor{sthlmOrange}{RGB}{221,74,44} 

	\definecolor{sthlmLightPurple}{RGB}{241,230,252} 
		
	\definecolor{sthlmPurple}{RGB}{93,35,125} 

	\definecolor{sthlmLightRed}{RGB}{254,222,237} 
		
	\definecolor{sthlmRed}{RGB}{196,0,100} 

	\definecolor{sthlmYellow}{RGB}{252,191,10} 

\usepackage{colortbl}%
  \newcommand{\myrowcolour}{\rowcolor[gray]{0.925}}
\usepackage{booktabs}

  \newcommand{\highest}[1]{\textcolor{sthlmRed}{\textbf{#1}}}%

\usepackage{xcolor}
\definecolor{pink}{rgb}{0.858, 0.188, 0.478}

\newcommand{\persComment}[3]{
  \ifmmode
  \text{\textcolor{#3}{[#2] #1}}
  \else
  \textcolor{#3}{[#2] \em #1}
  \fi
}
\newcommand{\new}[1]{
  #1
}
\newcommand{\old}[1]{
}

\icmltitlerunning{Solving Hierarchical Information-Sharing Dec-POMDPs}

\begin{document}

\twocolumn[
\icmltitle{Solving Hierarchical Information-Sharing Dec-POMDPs: \\ An Extensive-Form Game Approach}



\icmlsetsymbol{equal}{*}

\begin{icmlauthorlist}
\icmlauthor{Johan Peralez}{citi}
\icmlauthor{Aurélien Delage}{citi}
\icmlauthor{Olivier Buffet}{inria nancy}
\icmlauthor{Jilles S. Dibangoye}{groningen}
\end{icmlauthorlist}

\icmlaffiliation{citi}{Université de Lyon, INSA Lyon and Inria, CITI, F-69000 Lyon}
\icmlaffiliation{inria nancy}{Université de Lorraine, CNRS, INRIA, LORIA, F-54000 Nancy}
\icmlaffiliation{groningen}{
Bernoulli Institute, University of Groningen,
Nijenborgh 4, NL-9747 AG Groningen, Netherlands}

\icmlcorrespondingauthor{Jilles S. Dibangoye}{j.s.dibangoye@rug.nl}

\icmlkeywords{Machine Learning, ICML}

\vskip 0.3in
]



\printAffiliationsAndNotice 

\begin{abstract}
%
A recent theory shows that a multi-player decentralized partially observable Markov decision process can be transformed into an equivalent single-player game, enabling the application of \citeauthor{bellman}'s principle of optimality to solve the single-player game by breaking it down into single-stage subgames. However, this approach entangles the decision variables of all players at each single-stage subgame, resulting in backups with a double-exponential complexity. This paper demonstrates how to disentangle these decision variables while maintaining optimality under hierarchical information sharing, a prominent management style in our society. To achieve this, we apply the principle of optimality to solve any single-stage subgame by breaking it down further into smaller subgames, enabling us to make single-player decisions at a time. Our approach reveals that extensive-form games always exist with solutions to a single-stage subgame, significantly reducing time complexity. Our experimental results show that the algorithms leveraging these findings can scale up to much larger multi-player games without compromising optimality.
\end{abstract}

The multi-player decentralized partially observable Markov decision process (Dec-POMDP) is a general game-theoretic setting for decision-making by a team of collaborative players \citep{amato2013decentralized}. In this multi-player game, players must coordinate while they can neither see the actual state of the world nor explicitly share what they see or do with each other due to communication costs, latency, or noise. This so-called silent coordination dilemma provides a partial explanation of the worst-case complexity results---\ie infinite-horizon cases are undecidable, finite-horizon ones are NEXP-hard, and finding $\epsilon$-approximations remains hard  \citep{639686,rabinovich2003complexity}. Methods for Dec-POMDPs are split between local and global, each with strengths and weaknesses.

Local methods trade global optima, or $\epsilon$-approximations for weaker solution concepts, \eg local optima, Nash equilibria, or any arbitrary feasible solution. While they share core ideas with global methods, their primary focus is on solving relaxations of the original multi-player game, \eg independent planners reason in isolation, policy gradient targets first-order solutions of non-convex functions \citep{Tan:1997:MRL:284860.284934,peshkin2001learning,BonoDMP018}. Of particular attention, local methods using deep neural networks can apply effectively to virtually any non-critical application, \eg online services, logistics, or board games \citep{NIPS2017_68a97503,FoersterFANW18,RashidSWFFW18}. 

On the other hand, in many critical and high-stakes applications, \eg search and rescue, security, and healthcare, global methods can find solutions with the required theoretical guarantees, but scalability remains a significant issue. These algorithms recast the original multi-player game into an equivalent single-player one, which overcomes the silent coordination dilemma and allows the principle of optimality to apply. Intuitively, this principle decomposes the single-player game into single-stage subgames and solves them recursively. Yet, doing so, in return, virtually entangles decision variables of all players at each single-stage subgame, resulting in double-exponential complexity. Because they update decision variables of all players in sync at every single-stage subgame, even a single update can be prohibitively expensive \citep{szer2005optimal,NIPS2013_c9e1074f,nayyar2013decentralized,oliehoek2013sufficient}. To somewhat mitigate this burden, branch-and-bound search algorithms and mixed-integer linear programs were introduced, but the limitation remains \citep{OliehoekSDA10,DibangoyeMC09,DibangoyeABC13,Dibangoye2016}. In many cases, however, real-world environments contain significant structure that can be exploited \citep{amato2013decentralized}.

Indeed, several forms of structure have been investigated in the past---\eg dynamics independence \citep{BeckerZLG04,DibangoyeAD12}, weak-separability \citep{NairVTY05,Jillesaamas14}, and delayed information-sharing \citep{nayyar2010optimal}. Algorithms that use such structures can optimally solve structured multi-player games much faster than generic ones. This paper exploits HIS structure, a dominant management style in our society for corporations, governments, criminal enterprises, armies, and religions. This management style involves each player being aware of what its subordinate knows, and this knowledge is passed down the chain of command.  In other words, player $n$ at the top of the hierarchy knows all that player $n-1$ knows; player $n-2$ knows all that player $n-3$ knows, and so forth. Moreover, HIS is equivalent to one-sidedness when only two players are involved, which was previously recognized as a tractable structure for two-person partially observable stochastic games \citep{horak2017heuristic,HorBos-aaai19,hadfield2016cooperative,pmlr-v80-malik18a,pmlr-v119-xie20a}. Still, little is known about how HIS affects existing theory and algorithms. 

The main contribution of this paper is the proof that under the HIS assumption, perfect-information extensive-form games always exist with solutions to single-stage subgames, resulting in a significant reduction in time complexity.  When expressed as extensive-form games, one can optimize all decision variables in isolation while preserving optimality, resulting in an exponential drop in time complexity, hence generalizing to multiple players a similar property to that available under one-sidedness \citep{pmlr-v119-xie20a}. To show this result, we apply the principle of optimality to solve any single-stage subgame by breaking it down further into smaller subgames, enabling us to make one-player decisions at a time. In the resulting perfect-information extensive-form game, we exhibit concise representations of states and actions along with \citeauthor{bellman}'s optimality equations to solve the game. Finally, we present a point-based value-iteration algorithm for solving the original multi-player game leveraging HIS properties. Experiments show that algorithms exploiting these findings scale up to much larger multi-player games without compromising optimality.

\section{Background}
\label{sec:background}
This section presents state-of-the-art multi- and single-player formulations for Dec-POMDPs under HIS.
\vspace{-0.5cm}
\paragraph{Notations.} For integers $t_1\leq t_2$, $\kappa_{t_1: t_2}$ is a shorthand for $(\kappa_{t_1},\kappa_{t_1+1},\ldots,\kappa_{t_2})$. Let $\kappa_{t_1: t_2}$ be a complete vector, shorthands $\kappa_{t_1:}$ and $\kappa_{: t_2}$ denote suffix and prefix, respectively. \new{For two variables $a$ and $b$, we denote by $\delta_a^b$ the Kronecker delta, which is $1$ if $a$ equals $b$, and $0$ otherwise.}

\subsection{Multi-Player Formulation}

An $n$-player Dec-POMDP is given by tuple $M \doteq \langle n, X, U, Z, p,r, s_0,\gamma,\ell\rangle$,
where
 $X$ is a finite set of hidden states;
 $U^i$ is the finite actions set for player $i$, where $U = U^1\times \cdots\times U^n$ specifies the set of joint actions $u=(u^1,\ldots,u^n)$;
 $Z^i$ is the finite observation set for player $i$, where $Z = Z^1\times \cdots\times Z^n$ specifies the set of joint observations $z=(z^1,\ldots,z^n)$;
function $p\colon X\times U \to \triangle(X\times Z)$ describes a transition function with conditional probability distribution $p(y,z|x,u)$ defining the probability of transitioning from state $x$ to $y$ after taking joint action $u$ and seeing $z$;
 function $r\colon X\times U \to \mathbb{R}$ is a reward model with $r(x,u)$ being the immediate reward received after taking  joint action $u$ from state $x$; 
 $s_0$ is the initial state distribution, $\gamma$ is the discount factor, and $\ell$ is the number of stages.

In the remainder, we consider $M$ under the HIS assumption.
\new{That is, every player $0<i \leq n$ has instantaneous and cost-free
access to its subordinate's action $u^{i-1}_{\tau-1}$ and observation
$z^{i-1}_\tau$ at every stage $\tau$.
Consequently, there exists a function $\zeta^i$ that maps $z_\tau^i$ to
$\zeta^{i}(z_\tau^{i})=(u_{\tau-1}^{i-1},z_\tau^{i-1})$.}
Player $1$ is at the bottom of the hierarchy, \ie the player whose actions and observations are public to all other players, and player $n$ is at the top of the hierarchy, \ie the player that sees all actions and observations.
These characteristics are embodied in many real-world applications, including autonomous vehicle platooning, assembly line optimization, or railway traffic control.
Consider an autonomous vehicle platooning that relies on a leading vehicle followed by a group of autonomous vehicles; see Figure \ref{fig:example:v2x}.
Autonomous vehicles involved in platooning can exchange information between vehicles using Vehicle-to-Everything (V2X) communications in an HIS fashion \new{\citep{wang2015cooperative}}.
That is the total data transit from each autonomous vehicle $i$ to its following autonomous vehicle $i+1$.
The objective of platoon control is to determine the control input of the following autonomous vehicles so that all the vehicles move at the same speed while maintaining the desired distances between each pair of preceding and following vehicles.
Platooning constitutes an efficient technique for increasing road capacity, reducing fuel consumption, and enhancing driving safety and comfort.

\subsection{Limitations of Multi-Player Formulations} 

\begin{figure*}
\centering
\begin{tikzpicture}[->,>={Stealth[round]}, very thick, every edge quotes/.style = {auto, font=\footnotesize, sloped}]

\foreach \x/\y in {1/~n~,5/i\!\!+\!\!2,7/i\!\!+\!\!1,9/i,13/~1~} { 
  \node[token,minimum size=13pt,anchor=south,scale=1] at (\x,0.2) {$\y$} ;
  \node[] (wt\x) at (\x,-0.15) {\includegraphics[height=15mm]{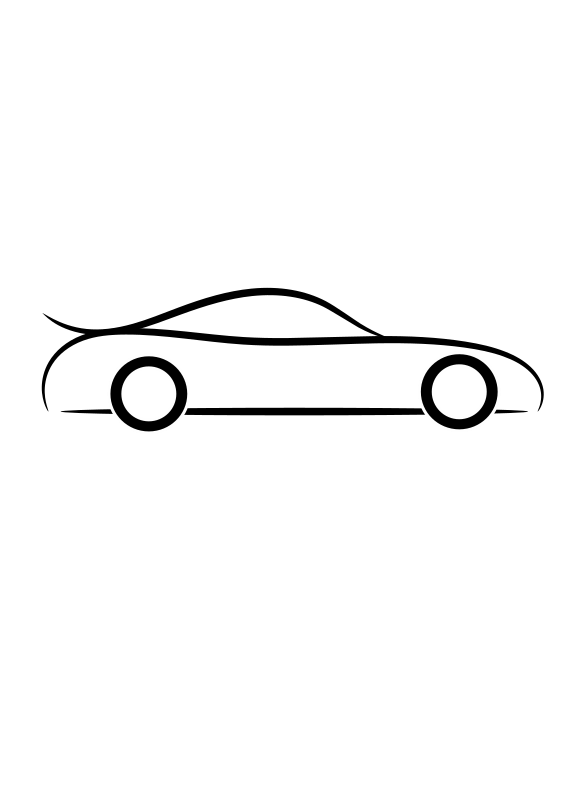}};
};

  \node[] (wtstart) at (3,-0.1) {$\circ\circ\circ$};
  \node[] (wtend) at (11,-0.1) {$\circ\circ\circ$};
  \path[->] (wt5) edge[sthlmGreen,bend right,dashed] (wt1);
  \path[->] (wt7) edge[sthlmGreen] (wt5);
  \path[->] (wt9) edge[sthlmGreen] (wt7);
  \path[->] (wt13) edge[sthlmGreen,bend right, dashed] (wt9);

\draw[] (0.5,-.25) -- (14,-.25);

\end{tikzpicture}
\vspace{-0.75cm}
\caption{V2X information transmitted to vehicles in the platooning.}
\label{fig:example:v2x}
\end{figure*}
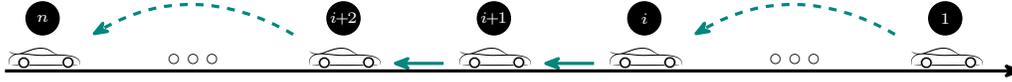

While HIS is a dominant management style in our society, little is known about how HIS affects existing theories and algorithms. In general, solving $M$ aims at finding optimal joint policy $a_{0:} \doteq (a_0,\ldots,a_{\ell-1}) $, \ie an $n$-tuple of sequences of private decision rules $a^i_{0:} \doteq (a^i_0,\ldots,a^i_{\ell-1})$, one per player. For each player $i$, private decision rule $a^i_\tau\colon o^i_\tau \mapsto u_\tau^i$ depends on $\tau$-step histories $o^i_\tau \doteq (u^i_{0:\tau-1},z^i_{1:\tau})$, with $0$-step private history being $o^i_0 \doteq \emptyset$. A joint policy is optimal if it maximizes the expected cumulative reward starting at initial state distribution $s_0$ onward and given by
$\upsilon_0^{a_{0:}}(s_0) \doteq \mathbb{E}_{(x_0,o_0)\sim \Pr\{\cdot|s_0, a_{0:} \}}\{ \alpha_0^{a_{0:}}(x_0,o_0)  \}$ where $\alpha_\tau^{a_{\tau:}}(x_\tau,o_\tau) \doteq \mathbb{E}_{(x_{\tau:\ell-1},u_{\tau:\ell-1})\sim \Pr\{\cdot|x_\tau,o_\tau, a_{\tau:} \}}\{ \textstyle{\sum_{t=\tau}^{\ell-1}}~\gamma^{\tau-t}\cdot r(x_t,u_t) \}$ for any game stage $\tau$.
Unfortunately, optimally solving $M$ in its multi-player formulation is non-trivial because of the silent coordination dilemma \citep{rabinovich2003complexity}. Indeed, no statistics on what the players see and do are sufficient to solve $M$ optimally. Private histories $o^i_\tau$ are not geared to perform policy evaluation, let alone policy ordering. In addition, joint histories $o_\tau \doteq (o^1_\tau,\ldots,o^n_\tau)$ cannot ensure policy disentanglement, \ie an individual policy per player.  To better understand this, notice that multi-player coordination is based on common ground, \ie knowledge, beliefs, and assumptions shared among players about the environment at each stage, making it possible to perform policy ordering and disentanglement.  Paradoxically, $M$ aims to coordinate agents without common ground, thus explaining the silent coordination dilemma. The motivation for a single-player reformulation is twofold: first, to provide the central planner with the common ground at the offline planning phase, which achieves policy ordering and disentanglement, then to ease the transfer of theories and algorithms from single- to multi-player formulations.

\subsection{Single-Player Reformulation}

The single-player reformulation describes $M$ from the perspective of an offline central planner \citep{szer2005optimal,nayyar2013decentralized,oliehoek2013sufficient,DibangoyeABC13,Dibangoye2016}. This planner reasons for all players in sync, prescribing a joint decision rule and receiving rewards and public observations.  Game $M$ lies in some underlying state and players have experienced a joint history at each plan-time stage. Unfortunately, the central planner can see neither the state nor the joint history. Yet, it can still prescribe to players what joint decision rule to follow based only upon the joint policy it has prescribed to players so far.  Upon executing the prescribed joint decision rule, the central planner receives the expected immediate reward and the next public observation, \ie the information of player $1$. This process follows at the next plan-time stage, but the game has another underlying state, and players are experiencing another joint history.  This process repeats until the number of stages is exhausted.  The summary of the history of prescribed joint decision rules and received public observations, \ie occupancy state, describes a Markov decision process.  Occupancy states proved to be common ground for coordinating players under the silent coordination dilemma, \ie occupancy states are sufficient statistics for optimal decision-making in Dec-POMDPs \citep{DibangoyeABC13,Dibangoye2016}.

Markov decision process  $M' \doteq \langle S, A, \pmb{T}, \pmb{R},s_0,\gamma,\ell\rangle$ \wrt $M$ consists of
 the occupancy-state space $S$, where occupancy states are conditional probability distribution over hidden states and joint histories;
the action space  $A$ prescribing joint decision rules;
the transition probability $\pmb{T}\colon S \times A\to\triangle(S)$, where 
$\pmb{T}(s_\tau,a_\tau,s_{\tau+1})\doteq \sum_{o,z} \delta_{\rho(s_\tau,a_\tau,z^1)}^{s_{\tau+1}} \sum_{x,y} {s}_\tau(x,o)  \cdot p(y,z|x,a_\tau(o))$,
 where the next occupancy state $s_{\tau+1}\doteq \rho(s_\tau,a_\tau,z^1_{\tau+1})$ follows from taking joint decision rule $a_\tau$ in occupancy state $s_\tau$ and then receiving public observation $z^1_{\tau+1}$, \ie for any arbitrary hidden state $y$ and joint history $(o,u,z)$, we have
${s}_{\tau+1}(y,(o,u,z)) \propto~ \sum_{x} {s}_\tau(x,o) \cdot \delta_{z_{\tau+1}^1}^{z^1}\cdot \delta_{a_\tau(o)}^{u}\cdot p(y,z|x,u)$;
and finally, $\pmb{R}\colon {S}  \times {A} \to \mathbb{R}$ is the expected immediate reward function, \ie
$\pmb{R}({s}_\tau,{a}_\tau) \doteq  \sum_{x,o} {s}_\tau(x,o) \cdot  r(x,{a}_\tau(o)) $.
Recasting the original multi-player game into an equivalent single-player one allows the principle of optimality to solve the single-player game by breaking it down into single-stage subgames and solving them recursively. Consequently, optimally solving $M$ aims at finding solutions $V_\tau^*(s_\tau)$ of single-stage subgame for every occupancy state $s_\tau$, \ie $V_\tau^*(s_\tau) = \max_{a_{\tau}} Q^*_\tau(s_\tau,a_\tau)$, where   $\textstyle Q^*_\tau(s_\tau,a_\tau) \doteq \pmb{R}(s_\tau,a_\tau) + \gamma\sum_{s_{\tau+1}}\pmb{T}(s_\tau,a_\tau,s_{\tau+1}) \cdot  V^*_{\tau+1}(s_{\tau+1})$ with boundary condition $V^*_\ell(\cdot) \doteq 0$. Each occupancy state $s_\tau$ has its corresponding single-stage subgame $G_{s_\tau}  \doteq \langle n, A, Q^*_\tau(s_\tau, \cdot) \rangle $, whose solution is $V^*_\tau(s_\tau)$.  Optimally solving single-stage subgame $G_{s_\tau}$ is significantly more efficient by leveraging the piecewise-linearity and convexity property of action-value functions $Q^*_\tau$. 

\begin{lemma}
\label{lem:action:value:pwlc}
For every game stage $\tau$, the optimal value function  $Q^*_\tau\colon S \times A \to \mathbb{R}$ is piecewise-linear and convex over occupancy states and joint decision rules. Alternatively, there exists a finite collection ${\mathcal{Q}}_\tau \subseteq \{ \beta^{a_{\tau+1:}}_\tau|  a_{\tau+1:}\in A_{\tau+1:}\}$  of action-value functions $\beta^{a_{\tau+1:}}_\tau$ under joint policy $a_{\tau+1:}$, such that: for  occupancy state $s_\tau$ and joint decision rule $a_\tau$,
$${Q}^*_\tau({s}_\tau, {a}_\tau) = \textstyle \max_{\beta_\tau\in {\mathcal{Q}}_\tau}~\mathbb{E}_{(x,o,u)\sim \Pr\{\cdot|{s}_\tau, {a}_\tau \}}\{ \beta_\tau(x,o,u)  \}$$
$$\beta^{a_{\tau+1:}}_\tau(x,o,u) =  r(x,u) +\textstyle\gamma\mathbb{E}_{(y,z)\sim p(\cdot|x,u)}\{  \alpha^{a_{\tau+1:}}_{t+1}(y,(o,z)) \}$$ with boundary condition $\alpha^{\cdot}_{\ell}(\cdot) =\beta^{\cdot}_{\ell}(\cdot) \doteq 0$.
\end{lemma}
Lemma \ref{lem:action:value:pwlc} allows us to optimally solve a single-stage subgame $G_{s_\tau}$ by taking the best among solutions of single-stage subgames $G_{s_\tau}^{\beta_\tau}  \doteq \langle n, A, Q_{\beta_\tau}(s_\tau, \cdot) \rangle$ induced by action-value function $\beta_\tau\in \mathcal{Q}_\tau$ under a fixed joint policy, where $Q_{\beta_\tau}(s_\tau, \cdot) \colon a_\tau \mapsto \mathbb{E}_{(x,o,u)\sim \Pr\{\cdot|s_\tau,a_\tau\}}\{\beta_\tau(x,o,u)\}$. In particular, the linearity of $Q_{\beta_\tau}(s_\tau, \cdot)$ over joint decision rules will play a crucial role in disentangling decision variables.

\subsection{Limitations of Single-Player Reformulations}
The single-player reformulation applies under HIS, but the curse of dimensionality restricts its scalability in the face of games with many players.  To better understand this, notice that the complexity of optimally solving a single-player reformulation depends on two operators: the point-based backup operator, which optimally solves single-stage subgame $G_{s_\tau}^{\beta_\tau}$, and the estimation operator, which updates all decision variables involved in the common ground, \ie occupancy states. In either case, the single-player reformulation is not geared to exploit HIS. State-of-the-art approaches to solving  $G_{s_\tau}^{\beta_\tau}$  perform either brute-force or implicit enumeration and evaluation of double-exponentially many joint decision rules \citep{OliehoekSDA10,DibangoyeMC09,DibangoyeABC13,Dibangoye2016}. This provides an intuitive explanation for the negative complexity results: optimally solving $G_{s_\tau}^{\beta_\tau}$ is NP-hard, and finding $\epsilon$-approximations remains hard \citep{Tsitsiklis84}. The estimation operator also suffers from the curse of dimensionality.  Indeed, the number of decision variables of all players in the common ground under the silent coordination dilemma grows exponentially with time and team size. In this paper, we investigate the following question.

\tikzstyle{mybox} = [draw=black, very thick, rectangle, rounded corners, inner ysep=5pt, inner xsep=5pt]
\vspace{10pt}
\begin{tikzpicture}
\node [mybox] (box){
\begin{minipage}{.96\linewidth}
\quad\emph{How can we improve the representations of common ground and Bellman optimality equations to scale-up point-based backup and estimation operators to optimally solving $G_{s_\tau}^{\beta_\tau}$and eventually $M'$ (resp. $M$) under HIS?}
\end{minipage}
};
\end{tikzpicture}

\section{Hierarchical Information Sharing}

This section explores the ramifications of HIS assumption in achieving an optimal solution for a single-stage subgame, specifically, $G_{s_\tau}^{\beta_\tau}$. 

\subsection{From Single-Stage to Extensive-Form Games}
%
While the principle of optimality allows us to break down the single-player reformulation $M'$ into smaller subgames $G_{s_\tau}^{\beta_\tau}$ per stage, an alternative approach is to segment single-stage subgames $G_{s_\tau}^{\beta_\tau}$ per player further. That allows the centralized planner to act sequentially for each player, starting from player $1$ up to player $n$. In addition, instead of choosing a decision rule for each player based on the current occupancy state and decision rules selected thus far, the planner can independently branch over each history that HIS makes available to the current player, without compromising optimality. A formal description of this process follows. 

 Starting from player $1$ at the bottom of the hierarchy, \cf Figure \ref{fig:ck:game:tree}, the planner chooses action $u^1_\tau$ according to its total available information $\varsigma^1_\tau= (s_\tau, o_\tau^1)$. It then moves to player $2$, the next player in the reversed order of the hierarchy, but now it randomly lands on total available information $\varsigma^2_\tau= (\varsigma_\tau^1, u^1_\tau, o_\tau^2)$ and chooses action $u^2_\tau$. The process continues until the planner reaches player $n$ at the top of the hierarchy, where it randomly lands on total available information $\varsigma^n_\tau= (\varsigma_\tau^{n-1}, u^{n-1}_\tau, o_\tau^n)$ and chooses action $u_\tau^n$ and receives expected rewards $R(\varsigma_\tau^n, u_\tau^n) \doteq \mathbb{E}_{x\sim \Pr\{\cdot|\varsigma_\tau^n,u_\tau^n\}}\{\beta_\tau(x,o,u)\}$ upon taking action $u_\tau^n$ in information state $\varsigma_\tau^n$. Upon acting sequentially for $i$ players, the total information available to the planner denoted $\varsigma_\tau^{i+1} \doteq (\varsigma_\tau^i, u^i_\tau, o_\tau^{i+1})$, is the current occupancy state $s_\tau$ of the single-stage subgame $G_{s_\tau}^{\beta_\tau}$ along with the sequence of actions that the planner selected and private histories that the planner received according to probability $T(\varsigma_\tau^{i+1}|\varsigma_\tau^i) \doteq  \Pr\{o_\tau^{i+1}|s_\tau, o_\tau^1,\ldots,o_\tau^i\} \cdot \delta_{\varsigma_\tau^i, u^i_\tau, o_\tau^{i+1}}^{\varsigma_\tau^{i+1}}$.

\begin{figure}[!htbp]
\centering
 \tikzstyle{every state}=[inner color= white,outer color= white,draw= black,text=black, drop shadow]
		\begin{tikzpicture}[ 
			scale=.75,
			ornode/.style={draw=sthlmRed, regular polygon, regular polygon sides=3, fill=sthlmRed!20}, 
			andnode/.style={draw=sthlmGreen, circle, fill=sthlmGreen!20}] 			

			\node[state, ornode] (q_0) at (0,0) {$\varsigma^1$};
			 
				\node[state, andnode, scale=.85] (q_1) at (-3, -1.5) {$u^1$}; 							
					\node[state, ornode, text=white, scale=.65] (q_11) at (-4.5, -3) {$\varsigma_\tau$};
					\node[scale=.6] at (-4.5, -3) {$\varsigma_{\square}^{2}$};
						\node (q_111) at (-5.25,-4.5) {};
						\node (q_112) at (-3.75,-4.5) {};
						\node (q_113) at (-4.2,-4) {};
						\node (q_114) at (-4.8,-4) {};
					\node[state, ornode, text=white, scale=.65] (q_12) at (-1.5, -3) {$\varsigma_\tau$}; 
					\node[scale=.6] at (-1.5, -3) {$\varsigma_{\textcolor{sthlmRed}{\square}}^{2}$}; 
						\node (q_121) at (-2.25,-4.5) {};
						\node (q_122) at (-0.75,-4.5) {};
						\node (q_123) at (-1.2,-4) {};
						\node (q_124) at (-1.8,-4) {};
					\node (q_13) at (-3.3, -2.5) {}; 
					\node (q_14) at (-2.7, -2.5) {}; 
					\node at (-3,-5) {$\cdots$};
					\node at (-3,-3) {$\cdots$};
					
				\node[state, andnode, scale=.85] (q_2) at (3, -1.5) {$u^1$}; 
					\node[state, ornode, text=white, scale=.65] (q_21) at (4.5, -3) {$\varsigma_\tau$}; 
					\node[scale=.6] at (4.5, -3) {$\varsigma_{\textcolor{sthlmRed}{\square}}^{2}$}; 
						\node (q_211) at (5.25,-4.5) {};
						\node (q_212) at (3.75,-4.5) {};
						\node (q_213) at (4.2,-4) {};
						\node (q_214) at (4.8,-4) {};
					\node[state, ornode, text=white, scale=.65] (q_22) at (1.5, -3) {$\varsigma_\tau$}; 
					\node[scale=.6]at (1.5, -3) {$\varsigma_{\square}^{2}$}; 
						\node (q_221) at (2.25,-4.5) {};
						\node (q_222) at (0.75,-4.5) {};
						\node (q_223) at (1.2,-4) {};
						\node (q_224) at (1.8,-4) {};
					\node (q_23) at (3.3, -2.5) {}; 
					\node (q_24) at (2.7, -2.5) {}; 
					\node at (3,-5) {$\cdots$};
					\node at (3,-3) {$\cdots$};

				\node[] (q_3) at (-.3, -1.5) {}; 
				\node[] (q_4) at (.3, -1.5) {}; 
				\node at (0,-1.5) {$\cdots$};
				\node at (0,-3) {$\cdots$};
				\node at (0,-5) {$\cdots$};

			\path[-, draw=black] (q_0) edge node[fill=white, scale=.7] {$u^1$} (q_1) 
								 edge node[fill=white, scale=.7] {$u^1$} (q_2) 
								 edge (q_3) 
								 edge (q_4) 
								 (q_1) edge node[fill=white, scale=.7] {$\square$} (q_11) %
								 	  edge node[fill=white, scale=.7] {$\textcolor{sthlmRed}{\square}$} (q_12) %
									  edge (q_13)
									  edge (q_14)
								 (q_2) edge node[fill=white, scale=.7] {$\textcolor{sthlmRed}{\square}$} (q_21) %
								 	  edge node[fill=white, scale=.7] {$\square$} (q_22) %
									  edge (q_23)
									  edge (q_24)
									  (q_11) edge (q_111)
									     edge (q_112)
									     edge (q_113)
									     edge (q_114)
									  (q_12) edge (q_121)
									     edge (q_122)
									     edge (q_123)
									     edge (q_124)
									  (q_21) edge (q_211)
									     edge (q_212)
									     edge (q_213)
									     edge (q_214)
									  (q_22) edge (q_221)
									     edge (q_222)
									     edge (q_223)
									     edge (q_224); 
		\end{tikzpicture} 	
		\caption{The search space for a single-stage subgame from a centralized planner acting sequentially one player at a time, illustrated as an AND/OR tree. OR nodes (triangle) represent alternative ways to solve $\bar{G}_{s_\tau}^{\beta_\tau}$. AND nodes (circle) represent subproblem alternatives to be solved. \textbf{Best viewed in color.}
		}
		\label{fig:ck:game:tree}
\end{figure}
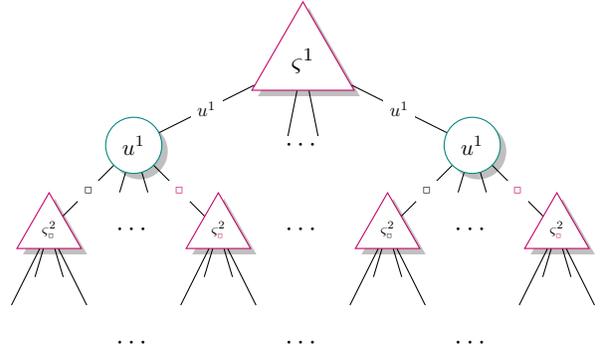

The total available information of the sequential-move central planner when solving a single-stage subgame describes a common-payoff perfect-information extensive-form game 
\citep{shoham2008multiagent}.
\begin{definition}
The common-payoff perfect-information extensive-form game\footnote{\new{
This definition differs from the formal definition of an extensive form game (EFG). To recover the standard EFG formalism, note that: i) agents act in sequence, ii) rewards are zero everywhere except at the leaves of the game tree, iii) stochastic transitions correspond to the presence of a chance player between two agents.
}}
\wrt $G_{s_\tau}^{\beta_\tau}$ is
a tuple $\bar{G}_{s_\tau}^{\beta_\tau} \doteq \langle n, \Sigma, \Psi, T, R \rangle$ where: $n$ is the number of players; $\Sigma$ is the set of nodes that occupancy state $s_\tau$ induces; $\Psi\colon \Sigma \to 2^{\cup_{i=1}^n U^i}$ is a function that specifies the allowed actions from each node $\varsigma\in \Sigma$; transition function $T\colon \Sigma \times (\cup_{i=1}^n U^i) \times \Sigma \to [0,1]$ specifies the probability of a successor node; reward function $R\colon \Sigma\times (\cup_{i=1}^n U^i) \to \mathbb{R}$ specifies the common payoff received upon taking an action in a node.
\end{definition}

 \subsection{Optimally Solving $G_{s_\tau}^{\beta_\tau}$ As $\bar{G}_{s_\tau}^{\beta_\tau}$ } \label{sec:optimally_solving_G}
 
Optimally solving a common-payoff perfect-information extensive-form game aims at finding the action-value functions $\beta^{1:n,*}_\tau$ mapping nodes and actions to optimal values.  Unlike the original single-stage subgame $G_{s_\tau}^{\beta_\tau}$, the perfect information extensive form game  $\bar{G}_{s_\tau}^{\beta_\tau}$ makes the HIS structure explicit. Every time the planner acts on behalf of a player, that player is perfectly informed about all the histories that have previously occurred---\ie all histories of its subordinates. Hence, the total information nodes include the actions the planner selected for the subordinates of the current player, along with the histories of its subordinates. Nonetheless, both games yield the same solution.
 
 \begin{theorem}
 \label{thm:bellman:equations}
 Any optimal solution for $\bar{G}_{s_\tau}^{\beta_\tau}$ is also an optimal solution for  $G_{s_\tau}^{\beta_\tau}$. Besides, the optimal action-value functions $\beta^{1:n,*}_\tau$ of $\bar{G}_{s_\tau}^{\beta_\tau}$ is the solution of the \citeauthor{bellman}'s optimality equations: at  any $i$,  $\varsigma_\tau^i$, and  $u_\tau^i$,
\begin{align*}
\beta^{i,*}_\tau(\varsigma_\tau^i, u^i_\tau) &=  \mathbb{E}_{\varsigma_\tau^{i+1}\sim T(\cdot|\varsigma_\tau^i, u^i_\tau) }\{  \max_{u^{i+1}_\tau}\beta_\tau^{i+1,*}(\varsigma_\tau^{i+1}, u_\tau^{i+1})\},
\end{align*}
with boundary condition  $\beta_\tau^{n,*}\colon (\varsigma_\tau^n,u^n_\tau) \mapsto R( \varsigma_\tau^n,u^n_\tau)$. Also, greedy decision rule $a^{i,*}_\tau$ for any player $i$ at $o^i_\tau$ is: 
\begin{align*}
a^{i,*}_\tau(o^i_\tau) &\textstyle\in \argmax_{u^i_\tau} \beta^{i,*}_\tau(\varsigma^i_\tau, u^i_\tau), 
\end{align*}
where $\varsigma^i_\tau \doteq   \langle s_\tau, o^{1:i}_\tau,a^{1:i-1,*}_\tau(o^{1:i-1}_\tau) \rangle$.
 \end{theorem}
 \begin{proof}
 The proof proceeds in two steps. First, it shows that the original game $G_{s_\tau}^{\beta_\tau}$ can alternatively be solved via a sequential-move central planner, which breaks $G_{s_\tau}^{\beta_\tau}$ down into smaller subgames $\langle G_{s_\tau,\emptyset}^{\beta_\tau}, G_{s_\tau, a^1_\tau}^{\beta_\tau}, \ldots, G_{s_\tau, a^{1:n-1}_\tau}^{\beta_\tau}\rangle$, one subgame per player. To this end, recall the goal of optimally solving $G_{s_\tau}^{\beta_\tau}$, \ie finding a joint decision rule which yields the highest performance index,
 $V_{\beta_\tau}(s_\tau) \doteq \max_{a_\tau}~ Q_{\beta_\tau}(s_\tau, a_\tau)$.
 The expansion of joint decision rule $a_\tau$ as a $n$-tuple of private decision rules $(a^1_\tau,a^2_\tau, \ldots, a^n_\tau)$ allows to rewite the objectif of $G_{s_\tau}^{\beta_\tau}$ as follows,  $V_{\beta_\tau}(s_\tau) =\max_{a^1_\tau}\max_{a^2_\tau} \ldots\max_{a^n_\tau}~ Q_{\beta_\tau}(s_\tau, a_\tau)$. Let $Q^i_{\beta_\tau}(s_\tau,\cdot) \colon a^{1:i}_\tau\mapsto \max_{a^{i+1:n}_\tau}~ Q_{\beta_\tau}(s_\tau, a_\tau)$ be a sequential action-value function.
Then, it follows that 
\begin{align*}
V_{\beta_\tau}(s_\tau) &= \max_{a^1_\tau}\max_{a^2_\tau} \ldots\max_{a^n_\tau}~ Q_{\beta_\tau}(s_\tau, a_\tau),\\
				  &=	\max_{a^1_\tau} ~\left[\max_{a^2_\tau} \ldots\max_{a^n_\tau}~ Q_{\beta_\tau}(s_\tau, a_\tau)\right],\\
				  &=\max_{a^1_\tau}~Q^1_{\beta_\tau}(s_\tau,a^1_\tau).
\end{align*}
Interestingly, for every player $i\in \{1,2,\ldots,n-1\}$, the action-value functions $Q^i_{\beta_\tau}(s_\tau,a^{1:i}_\tau)$ satisfy the following recursion
\begin{align*}
Q^i_{\beta_\tau}(s_\tau,a^{1:i}_\tau) &= \max_{a^{i+1}_\tau}\max_{a^{i+2}_\tau} \ldots\max_{a^n_\tau}~ Q_{\beta_\tau}(s_\tau, a_\tau),\\
				  &=	\max_{a^{i+1}_\tau} ~\left[\max_{a^{i+2}_\tau} \ldots\max_{a^n_\tau}~ Q_{\beta_\tau}(s_\tau, a_\tau)\right],\\
				  &=\max_{a^{i+1}_\tau}~Q^{i+1}_{\beta_\tau}(s_\tau,a^{1:i+1}_\tau),
\end{align*}
with boundary condition $Q^n_{\beta_\tau}(s_\tau,a_\tau) \doteq Q_{\beta_\tau}(s_\tau,a_\tau)$. For any arbitrary player $i\in \{2,3,\ldots, n\}$, define game $G^{\beta_\tau}_{s_\tau, a^{1:i-1}_\tau} \doteq \langle i, A^i, Q^i_{\beta_\tau}(s_\tau,a^{1:i-1}_\tau, \cdot)\rangle$ to be the subgame upon the sequential-move central planner selected decision rules $a^{1:i-1}_\tau$ starting in game $G_{s_\tau}^{\beta_\tau}$, with boundary condition $G^{\beta_\tau}_{s_\tau, \emptyset} \doteq \langle 1, A^1, Q^1_{\beta_\tau}(s_\tau,\cdot)\rangle$. Consequently, optimally solving the original game $G_{s_\tau}^{\beta_\tau}$ can be performed by optimally solving smaller subgames $\langle G_{s_\tau,\emptyset}^{\beta_\tau}, G_{s_\tau, a^1_\tau}^{\beta_\tau}, \ldots, G_{s_\tau, a^{1:n-1}_\tau}^{\beta_\tau}\rangle$, one subgame per player, recursvively. 

Next, we shall prove that the best decision rule in any arbitrary sequential-move subgame  $G^{\beta_\tau}_{s_\tau, a^{1:i-1}_\tau}$ depends on the current occupancy state $s_\tau$ along with previously selected decision rules $a^{1:i-1}_\tau$, only through the corresponding nodes $\varsigma_\tau^i \doteq (s_\tau, u^{1:i-1}_\tau, o^{1:i}_\tau)$ of the perfect information extensive form game $\bar{G}_{s_\tau}^{\beta_\tau}$. In other words, instead of selecting actions for all private histories of player $i$ in sync, one can choose the best action for each private history independently without compromising optimality. The proof of this statement proceeds by induction from player $n$ to player $1$. At player $n$, the greedy decision rule $\hat{a}^n_\tau$ satisfies the following:  
\begin{align*}
\hat{a}^n_\tau &\in\textstyle \argmax_{a^n_\tau}~Q^n_{\beta_\tau}(s_\tau, a_\tau),\\
		       &\in\textstyle \argmax_{a^n_\tau}~Q_{\beta_\tau}(s_\tau, a_\tau),\\
		       &\in\textstyle \argmax_{a^n_\tau}~\mathbb{E}_{(x,o,u)\sim \Pr\{\cdot| s_\tau, a_\tau\} }\{ \beta_\tau(x,o,u) \}.
\end{align*}  
Expanding over private histories of player $n$, we have that 
 \begin{align*}
\hat{a}^n_\tau(o^n_\tau) &\in\textstyle \argmax_{u^n_\tau}~\mathbb{E}_{(x,o,u)\sim \Pr\{\cdot| s_\tau, o^n_\tau, a_\tau\} }\{ \beta_\tau(x,o,u) \}.
\end{align*}  
Leveraging information available to player $n$ as provided by the HIS assumption, we know that the knowledge of private history $o_\tau^n$ implies the knowledge of histories of all other players  $o^{1:n-1}_\tau$, hence the joint history $o_\tau$, \ie
 \begin{align*}
\hat{a}^n_\tau(o^n_\tau) &\in\textstyle \argmax_{u^n_\tau}~\mathbb{E}_{x\sim \Pr\{\cdot| s_\tau, o_\tau, a_\tau\} }\{ \beta_\tau(x,o,u) \}.
\end{align*}  
 In addition, the knowledge of $o^{1:n-1}_\tau$ together with the decision rules $a^{1:n-1}_\tau$ the sequential-move central planner selected previously, makes it possible to access  node $\varsigma_\tau^n \doteq \langle s_\tau, o^{1:n}_\tau,a^{1:n-1}_\tau(o^{1:n-1}_\tau) \rangle$ such that:
 \begin{align*}
\hat{a}^n_\tau(o^n_\tau) &\in\textstyle \argmax_{u^n_\tau}~\beta_\tau^n(\varsigma_\tau^n,u^n_\tau),
\end{align*}
where $\beta_\tau^n\colon (\varsigma_\tau^n,u^n_\tau) \mapsto \mathbb{E}_{x\sim \Pr\{\cdot| \varsigma_\tau^n,u^n_\tau\} }\{ \beta_\tau(x,o,u) \}$, which proves the statement holds at player $n$.  Define  function $\alpha_\tau^n\colon \varsigma_\tau^n \mapsto \max_{u^n_\tau}\beta_\tau^n(\varsigma_\tau^n, u_\tau^n)$ at player $n$.  Notice that the value of the sequential-move subgame $G^{\beta_\tau}_{s_\tau, a^{1:n-1}_\tau}$ can be rewritten as follows:
\begin{align*}
Q^{n-1}_{\beta_\tau}(s_\tau, a_\tau^{1:n-1}) &= \max_{a^n_\tau} Q^n_{\beta_\tau}(s_\tau, a_\tau)\\
								   &=  \mathbb{E}_{\varsigma_\tau^n\sim \Pr\{\cdot| s_\tau,a_\tau^{1:n-1}\} }\{ \max_{u^n_\tau} \beta_\tau^n(\varsigma_\tau^n,u^n_\tau) \}\\
								   &=  \mathbb{E}_{\varsigma_\tau^n\sim \Pr\{\cdot| s_\tau,a_\tau^{1:n-1}\} }\{  \alpha_\tau^n(\varsigma_\tau^n) \}.
\end{align*}

Suppose the statement holds for any player $i>1$, with greedy decision rule $\hat{a}^i_\tau(o^i_\tau) \in \argmax_{u^i_\tau}~\beta_\tau^i(\varsigma_\tau^i, u^i_\tau)$. Define  function $\alpha_\tau^i\colon \varsigma_\tau^i \mapsto \max_{u^i_\tau}\beta_\tau^i(\varsigma_\tau^i, u_\tau^i)$ at player $i$. Also, the value of the sequential-move subgame $G^{\beta_\tau}_{s_\tau, a^{1:i-1}_\tau}$ can be rewritten by expanding over the sequential-move nodes $\varsigma_\tau^i \doteq  \langle s_\tau, o^{1:i}_\tau,a^{1:i-1}_\tau(o^{1:i-1}_\tau) \rangle$, \ie
\begin{align*}
Q^{i-1}_{\beta_\tau}(s_\tau, a_\tau^{1:i-1}) 
								   &=  \mathbb{E}_{\varsigma_\tau^i\sim \Pr\{\cdot| s_\tau,a_\tau^{1:i-1}\} }\{  \alpha_\tau^i(\varsigma_\tau^i) \}.
\end{align*}

 We are now ready to prove the statement also holds at player $i-1$. From the sequential-move central planner's viewpoint, decision rule $\hat{a}_\tau^{i-1}$ satisfies the following expression:
\begin{align*}
\hat{a}^{i-1}_\tau &\textstyle\in \argmax_{a^{i-1}_\tau}~Q^{i-1}_{\beta_\tau}(s_\tau, a^{1:i-1}_\tau),\\
		       &\textstyle\in \argmax_{a^{i-1}_\tau} \mathbb{E}_{\varsigma_\tau^i\sim \Pr\{\cdot| s_\tau,a_\tau^{1:i-1}\} }\{  \alpha_\tau^i(\varsigma_\tau^i) \}.
\end{align*}  
%
%
Similarly to player $n$, the knowledge of $o^{1:i-1}_\tau$ together with the decision rules $a^{1:i-2}_\tau$ the sequential-move central planner selected previously, makes it possible to access node $\varsigma_\tau^{i-1}\doteq \langle s_\tau, o^{1:i-1}_\tau,a^{1:i-2}_\tau(o^{1:i-2}_\tau) \rangle$ such that: 
\begin{align*}
\hat{a}^{i-1}_\tau(o^{i-1}_\tau) &\textstyle\in \argmax_{u^{i-1}_\tau}~ \beta_\tau^{i-1}(\varsigma_\tau^{i-1}, u_\tau^{i-1}),
\end{align*}  
where $\beta_\tau^{i-1}\colon (\varsigma_\tau^{i-1}, u_\tau^{i-1}) \mapsto \mathbb{E}_{\varsigma_\tau^i\sim \Pr\{\cdot| \varsigma_\tau^{i-1}, u_\tau^{i-1}\} }\{ \alpha_\tau^i(\varsigma_\tau^i) \}$, which proves the statement holds at player $i-1$. Define  function $\alpha_\tau^{i-1}\colon \varsigma_\tau^{i-1} \mapsto \max_{u^{i-1}_\tau}\beta_\tau^{i-1}(\varsigma_\tau^{i-1}, u_\tau^{i-1})$ at player $i-1$.  Consequently, the value of the sequential-move subgame $G^{\beta_\tau}_{s_\tau, \emptyset}$ can be rewritten by expanding over the sequential-move nodes $\varsigma_\tau^1\doteq  \langle s_\tau, o^1_\tau \rangle$, \ie
\begin{align*}
V_{\beta_\tau}(s_\tau) &=  \mathbb{E}_{\varsigma_\tau^1\sim \Pr\{\cdot| s_\tau\} }\{  \alpha_\tau^1(\varsigma_\tau^1) \}.
\end{align*}
The value of a cooperative game being unique, we know the optimal solution for $\bar{G}_{s_\tau}^{\beta_\tau}$ is also an optimal solution for $G_{s_\tau}^{\beta_\tau}$. In demonstrating this statement, we also exhibited \citeauthor{bellman}'s optimality equations, providing the solution of the perfect-information extensive-form game $\bar{G}_{s_\tau}^{\beta_\tau}$, \ie at any player $i$, node $\varsigma_\tau^i$, and action $u_\tau^i$,
\begin{align*}
\beta^{i,*}_\tau(\varsigma_\tau^i, u^i_\tau) &=  \mathbb{E}_{\varsigma_\tau^{i+1}\sim T(\cdot|\varsigma_\tau^i, u^i_\tau) }\{  \max_{u^{i+1}_\tau}\beta_\tau^{i+1,*}(\varsigma_\tau^{i+1}, u_\tau^{i+1})\},
\end{align*}
with boundary condition  $\beta_\tau^{n,*}\colon (\varsigma_\tau^n,u^n_\tau) \mapsto R( \varsigma_\tau^n,u^n_\tau)$.
Which ends the proof.
 \end{proof}
Theorem \ref{thm:bellman:equations} introduces \citeauthor{bellman}'s optimality equations that enable us to find a greedy joint decision at single-stage subgame $G_{s_\tau}^{\beta_\tau}$ by solving the corresponding extensive-form game $\bar{G}_{s_\tau}^{\beta_\tau}$. 
It proceeds in two phases. From player $n$ at the top of the hierarchy to player $1$ at the bottom, a backward pass computes optimal action-values $\beta^{i,*}_\tau(\varsigma_\tau^i, u_\tau^i)$ for each player $i$, each node 
$\varsigma^i_\tau$, and each action $u_\tau^i$. Then, from player $1$ at the bottom of the hierarchy to player $n$ at the top, a forward pass selects a greedy decision rule independently for each player $i$, and each node $\varsigma_\tau^i$. 
This backward induction algorithm requires a linear time complexity with the number of players, nodes, and actions $\mathbf{O}(n |\Sigma| |U^*|)$ instead of double exponential $ \mathbf{O}( |O^*|^{|U^*|^n})$ where $O^* \doteq \argmax_{O^i} |O^i|$ with $O^i$ being the set of reachable histories of player $i$ in $s_\tau$ and $U^* \doteq \argmax_{U^i} |U^i|$. A careful reader would notice that the linearity of $\beta_\tau$ over occupancy states and joint decision rules is key in demonstrating Theorem \ref{thm:bellman:equations}. 

\subsection{Nested-Occupancy States}

Upon inspection of perfect-information extensive-form game $\bar{G}_{s_\tau}^{\beta_\tau}$, one can see that despite the polynomial-time complexity of the point-based backup, $\bar{G}_{s_\tau}^{\beta_\tau}$ may contain a significant number of nodes. That is because nodes in $\bar{G}_{s_\tau}^{\beta_\tau}$ provide total information available to the planner at any player $i$---\ie $\varsigma_\tau^i \doteq \langle s_\tau, o^i_\tau, u^{:i-1}_\tau \rangle$, which may result in redundant and unnecessary computations. To address this challenge, we propose the introduction of a statistic referred to as \emph{nested-occupancy state} that we shall maintain in place of the total information available to the planner.

At player $i$, a nested-occupancy state $s^i_\tau\doteq (b^i_\tau, o^i_\tau, u^{:i-1}_\tau)$ consists of a private history $o^i_\tau$ of player $i$, the actions of its subordinates $u^{:i-1}_\tau$, and a nested-belief state $b^i_\tau$. Besides, the nested-belief state $b^i_\tau$ at player $i$ is a posterior distribution over histories $o^{i+1}_\tau$ and nested-belief states $b^{i+1}_\tau$ of the immediate superior player $i+1$. This distribution is conditional on the total data available to the planner at player $i$, i.e., $b^i_\tau(o^{i+1}_\tau,b^{i+1}_\tau) \doteq \Pr\{ o^{i+1}_\tau, b^{i+1}_\tau| \varsigma^i_\tau\}$, for any histories $o^{i+1}_\tau$ and nested-belief states $b^{i+1}_\tau$; with boundary condition $b^n_\tau(x_\tau)\doteq \Pr\{x_\tau| \varsigma^n_\tau\}$, for any hidden state $x_\tau$.  Interestingly, the nested-occupancy state has many important properties. First, it is a sufficient statistic for optimally solving $\bar{G}_{s_\tau}^{\beta_\tau}$.

\begin{theorem}[Proof in Appendix \ref{appendix:thm:sufficiency}]
\label{thm:sufficiency}
At player $i$, the nested-occupancy state $s^i_\tau$ is a sufficient statistic of  the total data $\varsigma_\tau^i$ available to the planner for optimally solving the perfect-information extensive-form game $\bar{G}_{s_\tau}^{\beta_\tau}$.
\end{theorem}
Theorem \ref{thm:sufficiency} suggests using a nested-occupancy state as an alternative to the total data available to the planner without compromising optimality. This statistic facilitates the aggregation of histories of a player that convey the same information about the game, thus effectively reducing the dimensionality of the game. Prior to delving further, it is necessary to introduce three equivalence relations. First, two nested-occupancy states at player $i$, represented as $s_\tau^{i,\textcolor{gray}{\bullet}}\doteq (b^{i,\textcolor{gray}{\bullet}}_\tau, o^{i,\textcolor{gray}{\bullet}}_\tau, u^{:i-1,\textcolor{gray}{\bullet}}_\tau)$ and $s_\tau^{i,\circ}\doteq (b^{i,\circ}_\tau, o^{i,\circ}_\tau, u^{:i-1,\circ}_\tau)$, are considered $\mathscr{B}_1$-equivalent if they differ only through their histories, \ie whenever $(b^{i,\circ}_\tau,u^{:i-1,\circ}_\tau)=(b^{i,\textcolor{gray}{\bullet}}_\tau,u^{:i-1,\textcolor{gray}{\bullet}}_\tau)$ then $s_\tau^{i,\textcolor{gray}{\bullet}}\sim_{\mathscr{B}_1} s_\tau^{i,\circ}$. Similarly, they are considered $\mathscr{B}_2$-equivalent if they share the same nested-belief state and histories of subordinates, \ie whenever $ (b^{i,\circ}_\tau,o^{i-1,\circ}_\tau)=(b^{i,\textcolor{gray}{\bullet}}_\tau,o^{i-1,\textcolor{gray}{\bullet}}_\tau)$ then $s_\tau^{i,\textcolor{gray}{\bullet}}\sim_{\mathscr{B}_2} s_\tau^{i,\circ}$. Last, two private histories are considered $\mathscr{P}$-equivalent if their optimal actions or, more generally, policies are interchangeable.
\begin{theorem}[Proof in Appendix \ref{appendix:thm:compression}]
\label{thm:compression}
Let $\bar{G}_{s_\tau}^{\beta_\tau}$ be a perfect-information extensive-form game. Let $s_\tau^{i,\textcolor{gray}{\bullet}}$ and $s_\tau^{i,\textcolor{gray}{\circ}}$ be two nested-occupancy states induced by occupancy state $s_\tau$ at player $i$. The following properties hold.
\begin{compactenum}
\item If $s_\tau^{i,\textcolor{gray}{\bullet}}\sim_{\mathscr{B}_1} s_\tau^{i,\circ}$ then $\beta^{i,*}_\tau(s_\tau^{i,\textcolor{gray}{\bullet}},u^i_\tau) = \beta^{i,*}_\tau(s_\tau^{i,\textcolor{gray}{\circ}},u^i_\tau)$.
\item If $s_\tau^{i,\textcolor{gray}{\bullet}}\sim_{\mathscr{B}_2} s_\tau^{i,\circ}$ then $o_\tau^{i,\textcolor{gray}{\bullet}} \sim_{\mathscr{P}} o_\tau^{i,\textcolor{gray}{\circ}}$.
\end{compactenum}
\end{theorem}
Theorem \ref{thm:compression} establishes that two $\mathscr{B}_1$-equivalent nested-occupancy states have the same optimal actions, resulting in significant computational savings. Moreover, it showcases that two $\mathscr{B}_2$-equivalent nested-occupancy states have their corresponding histories following the same policy. This insight allows for compact occupancy states, wherein only one history per equivalent class is retained, leading to faster estimations. Similarly, \citet{Dibangoye2016} employed compact occupancy states while utilizing the complete distribution over hidden states and histories of all teammates to determine when two histories are equivalent. Our equivalence relations, however, are based on nested belief states that are more concise than occupancy states, resulting in a more aggressive compression. At player $n$, for instance, the planner groups together histories that share the belief state, and this process continues down the hierarchy.

\section{Near-Optimally Solving $M'$ Under HIS} \label{sec:solving_subgame}

This section
adapts the point-based value-iteration (PBVI) algorithm \citep{pineau2003point} to compute $\epsilon$-optimal joint policy for  $M'$ (resp. $M$) under HIS starting at initial state distribution $s_0$ for planning horizon $\ell$. We chose the PBVI algorithm because it leverages the linear functions $\beta_\tau$ involved in the optimal value function. Besides, it is guaranteed to find near-optimal solutions asymptotically. Notice that algorithms that do not leverage the linear functions $\beta_\tau$, \eg feature-based heuristic search value iteration \citep{DibangoyeABC13,Dibangoye2016}, cannot benefit from our findings.

PBVI, \cf Algorithm \ref{pbvi:his:dec:pomdp} in Appendix \ref{sec:pbvi:algorithm}, has two main parts for solving $M'$ (resp. $M$) under HIS. First, it bounds the size of the value function at each stage $\tau$ of the game by representing the value only at a finite, reachable occupancy subset $\tilde{S}_\tau$.  Next, it optimizes the value function represented as a collection $\mathcal{V}_\tau$ at each stage $\tau$ using point-based backup, \ie at any stage $\tau$, 
$\mathcal{V}_\tau = \{ \mathtt{backup}(s_\tau, \mathcal{V}_{\tau+1})\colon s_\tau\in \tilde{S}_\tau \}$,
where backups are executed in no particular order, \ie
$\mathtt{backup}(s_\tau, \mathcal{V}_{\tau+1}) = \argmax_{\alpha_\tau^{a_{\tau:}}\colon a_\tau\in A, \alpha_\tau^{a_{\tau+1:}}\in \mathcal{V}_\tau} ~Q_{\beta_\tau^{a_{\tau+1:}}} (s_\tau, a_\tau)$.
Each iteration traverses occupancy-state subsets bottom up.  This iterative process repeats until convergence or until a budget, \eg CPU time, memory, or number of iterations, has been exhausted. The algorithm adds supplemental points into occupancy subsets to improve the value functions further. It selects candidate points using a portfolio of exploration strategies, including random explorations and greedy \wrt underlying (PO)MDP value functions. For every stage $\tau$, the algorithm adds only candidate points beyond a certain distance from the occupancy subset $\tilde{S}_\tau$ to create a new occupancy-state set $\tilde{S}_{\tau+1}$.

For any arbitrary occupancy-state subsets $\tilde{S}_{0:}$,  PBVI  produces a value $\upsilon_0(s_0)$. 
The error between $\upsilon_0(s_0)$ and $\upsilon^*_{0}(s_0)$ is bounded. The bound depends on how $\tilde{S}_{0:}$ samples the entire occupancy-state space; with denser sampling, the estimate $\upsilon_0(s_0)$ converges to $\upsilon^*_0(s_0)$. The remainder of this section states and proves our error bound. \new{It is also shown that the PBVI algorithm under HIS allows an exponential decrease in time complexity over standard versions of PBVI.}

Define the density $\delta_{ \tilde{S}_{0:}}$ to be the maximum distance from any \old{legal} \new{reachable} occupancy state to subsets $\tilde{S}_{0:}$. More precisely, $\delta_{ \tilde{S}_{0:}} \doteq \max_{\tau\in \llbracket 0:\ell-1\rrbracket}\max_{s\in S_\tau}\min_{s'\in \tilde{S}_\tau} \|s-s'\|_1$. Define a positive scalar $c$ such that $\|r(\cdot,\cdot)\|_\infty \leq c$.
\begin{theorem}[Proof in Appendix \ref{appendix:thm:error:bound}]
\label{thm:error:bound}
For any occupancy subsets $\tilde{S}_{0:}$, the error of the PBVI algorithm is bounded by $$\upsilon^*_0(s_0) - \upsilon_0(s_0) \leq  2c \delta_{ \tilde{S}_{0:}} \frac{1+\ell \gamma^{\ell+1}-(\ell+1) \gamma^\ell }{(1-\gamma)^2}.$$ 
\end{theorem}
It is worth noticing that whenever $\ell $ goes to infinity, our bound meets that from \citet{pineau2003point} for infinite-horizon partially observable Markov decision processes.

\new{
\begin{theorem} \label{th:complexity}
Let $|\tilde{S}^*| = \max_{t\in {0,1,\ldots,\ell-1}} |\tilde{S}_t|$ be the maximum size of the selected spaces of occupancy states.
The complexity of the PBVI algorithm under HIS is about $\pmb{O}\left(n \ell |\tilde{S}^*|^2 |Z^*|^{n\ell} |U^*|^{1+n(\ell+1)}\right)$.
\end{theorem}
\begin{proof}
    As stated in \mbox{\Cref{sec:optimally_solving_G}}, the complexity of solving a single-stage subgame $G_{s_\tau}^{\beta_\tau}$ is about $\pmb{O}(n |\Sigma||U^*|)$, where $|\Sigma|$ is the size of the extensive-form game. Since the set of reachable histories for each player is bounded in size by $(|Z^*||U^*|)^{n\ell}$, we get $|\Sigma| \leq (|Z^*||U^*|)^{n\ell} |U^*|^n$. \\
    At stage $\tau$ a subgame is solved for each $s_\tau \in \tilde S_{\tau}$ and each $\beta_\tau$ (obtained from $\tilde S_{\tau + 1})$.
    Thus, at most $\ell |\tilde S^*|^2$ subgames are solved on the whole horizon. \\
    As a consequence, the total complexity is about 
    $\mathbf{O}(\ell |\tilde{S}^*|^2 (|Z^*|  |U^*|)^{n\ell}  |U^*|^n n|U^*|) = \mathbf{O}(n \ell |\tilde{S}^*|^2 |Z^*|^{n\ell}  |U^*|^{1+n(\ell+1)})$.
\end{proof}
As a comparison, the number of joint decision rules at step-time $\tau$ is bounded by $|U^*|^{(|Z^*||U^*|)^{n \ell}}$.
Thus, using similar reasoning as in the proof of \mbox{\Cref{th:complexity}},
the complexity of the vanilla PBVI algorithm on the single-agent reformulation is about $\pmb{O}\left(\ell |\tilde{S}^*|^2 |U^*|^{(|Z^*||U^*|)^{n \ell}}\right)$
%
}


\section{Experiments} \label{sec:experiments}

This section presents the outcomes of our experiments, which were carried out to juxtapose our findings with the leading-edge theory employed in global methods, encompassing the utilization of the PBVI algorithm as a standard algorithmic scheme. Our analysis involves three variants of the PBVI algorithm, namely PBVI$^{enum}$, PBVI$^{milp}$, and hPBVI, each employing distinct methods of performing point-based backups. PBVI$^{enum}$ relies on brute-force enumeration of joint decision rules. At the same time, PBVI$^{milp}$ utilizes mixed-integer linear programs (MILPs) for implicit enumeration following the state-of-art approach for general Dec-POMDPs \citep{Dibangoye2016}. \new{In contrast, hPBVI leverages the subgame solving methods described above.} We used ILOG CPLEX Optimization Studio to solve the MILPs. Finally, hPBVI incorporates our findings to facilitate point-based backups under hierarchical information sharing.  Global methods are not designed to scale up with players. To present a comprehensive view, we have also compared our results against local policy- and value-based methods, \ie \old{asynchronous} \new{advantage} actor-critic (A2C) \cite{Konda1999ActorCriticA} and independent $Q$-learning (IQL) \cite{Tan:1997:MRL:284860.284934}, respectively. The experiments were executed on an Ubuntu machine with 32GB of available RAM and a 2.5GHz processor, utilizing only one core, with a time limit of 30 minutes.

We have comprehensively assessed various algorithms using several two-player benchmarks sourced from academic literature, available at \url{masplan.org}. These benchmarks encompass mabc, recycling, grid3x3, boxpushing, mars, and tiger. To enable a comparison of multiple players, we have also introduced the multi-player variants of these benchmarks. Please refer to Appendix \ref{sec:multi:player:benchmarks}  for a detailed definition of these multi-player benchmarks. 

\newcommand{\oot}{\multicolumn{2}{c}{\sc oot}}

\begin{table}
\label{table:results}
\resizebox{.5\textwidth}{!}{ 
\begin{tabular}{@{}l rr rr rr rr rr } 

\toprule%
 \centering%
 & \multicolumn{2}{c}{{{\bfseries hPBVI}}}
 & \multicolumn{2}{c}{{{\bfseries PBVI$^{milp}$}}}
 & \multicolumn{2}{c}{{{\bfseries PBVI$^{enum}$}}}
 & \multicolumn{2}{c}{{{\bfseries A2C}}}
 & \multicolumn{2}{c}{{{\bfseries IQL}}}
 \\

\cmidrule[0.4pt](r{0.125em}){1-1}%
\cmidrule[0.4pt](lr{0.125em}){2-3}%
\cmidrule[0.4pt](lr{0.125em}){4-5}%
\cmidrule[0.4pt](lr{0.125em}){6-7}%
\cmidrule[0.4pt](lr{0.125em}){8-9}%
\cmidrule[0.4pt](lr{0.125em}){10-11}%

tiger(2)& \highest{0.18} & \highest{112.50} & 1.63 & 91.81 & \oot & -- & 95.73 & -- & 80.15  \\
\myrowcolour tiger(3) & \highest{1.05} & \highest{262.50} & 141.72 & 218.81 & \oot & -- & 167.16 & -- & 255.99  \\
tiger(4) & \highest{6.28} & \highest{393.75} &  \oot & \oot & -- & 207.70 & -- & 218.47  \\
\myrowcolour tiger(6) & \highest{912.63} & \highest{483.78} & \oot & \oot  & -- & 200.96 & -- & -129.51  \\

\myrowcolour recycling(2) & \highest{0.02} & \highest{93.73} & {0.78} & \highest{93.73} & {0.04} & \highest{93.73} & -- & 93.34 & -- & 93.02  \\
recycling(3)  & \highest{0.05} & \highest{252.83} & 19.28 & \highest{252.83} & 143.59 & 247.80 & -- & 142.00 & -- & 129.57  \\
\myrowcolour recycling(4)  & \highest{0.19} & \highest{310.07} &  835.96 & 283.05 & \oot  & -- & 181.25 & -- & 153.03  \\
recycling(6)  & \highest{1.91} & \highest{459.78} &  \oot & \oot & -- & 186.11
& -- & 197.93  \\
\myrowcolour recycling(8)  & \highest{138.28} & \highest{600.00} &  \oot & \oot  & -- & 126.19 & -- & 244.02 \\

mabc(2) & 0.05 & \highest{27.42} & 0.15 & 27.40 & \highest{0.04} & \highest{27.42} & -- & 27.18 & -- & 27.2  \\

\myrowcolour%
mabc(3) & \highest{0.03} & \highest{23.24} & 1.21 & \highest{23.24} & 0.65 & \highest{23.24} & -- & 23.27 & -- & 23.24  \\

mabc(4) & \highest{0.07} & \highest{24.94}
& 223.26 & \highest{24.94} &  \oot & -- & 24.36 & -- & \highest{24.94}   \\

\myrowcolour%
mabc(7) & \highest{1.66} & \highest{27.25} & \oot & \oot & -- & 16.72 & -- & 26.82  \\

mabc(10) & \highest{139.82} & \highest{27.75} & \oot & \oot & -- & 12.25 & -- & 24.84  \\

\myrowcolour%
grid3x3(2) & \highest{0.61} & \highest{24.44} & 1329.33 & 24.33 & \oot & -- & 22.93 & --
& 24.35  \\

grid3x3(3) & \highest{65.43} & \highest{28.16} & \oot & \oot & -- & 27.92 & --
& \highest{28.16}  \\

\myrowcolour%
mars(2) & \highest{0.28} & \highest{84.33} &
248.61 & 76.15 &  \oot & -- & 43.20 & -- & 52.86   \\

boxpushing(2) & \highest{0.66} & \highest{675.46} &
24.58 & 576.30 & \oot  & -- & 180.11 & -- & 614.6  \\

\bottomrule

\end{tabular}
}
\vspace{-.25cm}
\caption{Snapshot of empirical results, \cf Appendix \ref{sec:experimental:results}. For each game($n$) and algorithm, we report time (in seconds) per backup and the best value for horizon $\ell=30$. {\sc oot} means time limit of 30 minutes has been exceeded and '--' is not applicable.  }

\end{table}

Our study aimed to assess the reduction in complexity achieved by point-based backups and its effect on solving larger multi-player games. Our findings show that hPBVI performs point-based backups significantly faster than other methods, which enables it to scale up to larger teams, as illustrated in Table \ref{table:results}. Specifically, hPBVI was able to perform point-based backups for up to 10 players in about 139.82 seconds in $\mathtt{mabc(10)}$ at $\ell=30$, while PBVI$^{enum}$ ran out of time for 4 players, and PBVI$^{milp}$ for 5 players. Additionally, hPBVI converges faster than PBVI$^{enum}$ and PBVI$^{milp}$ in 2- to 3-player domains. For example, hPBVI can converge in under 1 second in $\mathtt{grid3x3(2)}$ at $\ell=30$, while PBVI$^{milp}$ takes about 1329.33 seconds, not to mention PBVI$^{enum}$. Our results in Table \ref{table:results} demonstrate that hPBVI can scale up to larger teams of players where neither PBVI$^{milp}$ nor PBVI$^{enum}$ can. 
\new{Figure \ref{fig:backupTimeN:recycling_main} illustrates the capacity of hPBVI to address larger problems when compared to standard PBVI algorithms (a more extensive comparison of computational times is proposed in Appendix \ref{sec:experimental:results}).}

Local methods A2C  and IQL do scale up to larger teams as expected.  Surprisingly, they perform very well on certain domains with weakly coupled players, as shown in $\mathtt{mabc(4)}$ and $\mathtt{grid3x3(3)}$, \cf Table \ref{table:results}. However, hPBVI always performs better  A2C and IQL on all benchmarks except $\mathtt{mabc(4)}$ and $\mathtt{grid3x3(3)}$, which exhibit local behaviors that are global optimal solutions. Moreover, it  converges faster than A2C and IQL on all tested benchmarks,
\old{Figures \ref{fig:anytimeCurves_tiger},\ref{fig:anytimeCurves_recycling},\ref{fig:anytimeCurves_mabc},\ref{fig:anytimeCurves_grid3x3} report anytime 
performances in Appendix \ref{sec:experimental:results}.}
\new{\Cref{fig:anytimeCurves_recycling_main} illustrates anytime performances for the recycling problem 
(\Cref{fig:anytimeCurves_tiger,fig:anytimeCurves_recycling,fig:anytimeCurves_mabc,fig:anytimeCurves_grid3x3} in Appendix \ref{sec:experimental:results} provide more detailed results for each benchmark).}
Although this observation goes beyond our original goal, it provides encouraging insights when comparing local against global methods over teams of medium sizes. Nonetheless, we caution readers against drawing general conclusions from this observation, as different local methods may yield different local optima and convergence rates. 

\begin{figure}[ht]	
	\label{fig:backupTimeN:recycling_main}
    \includegraphics[width=1\columnwidth]{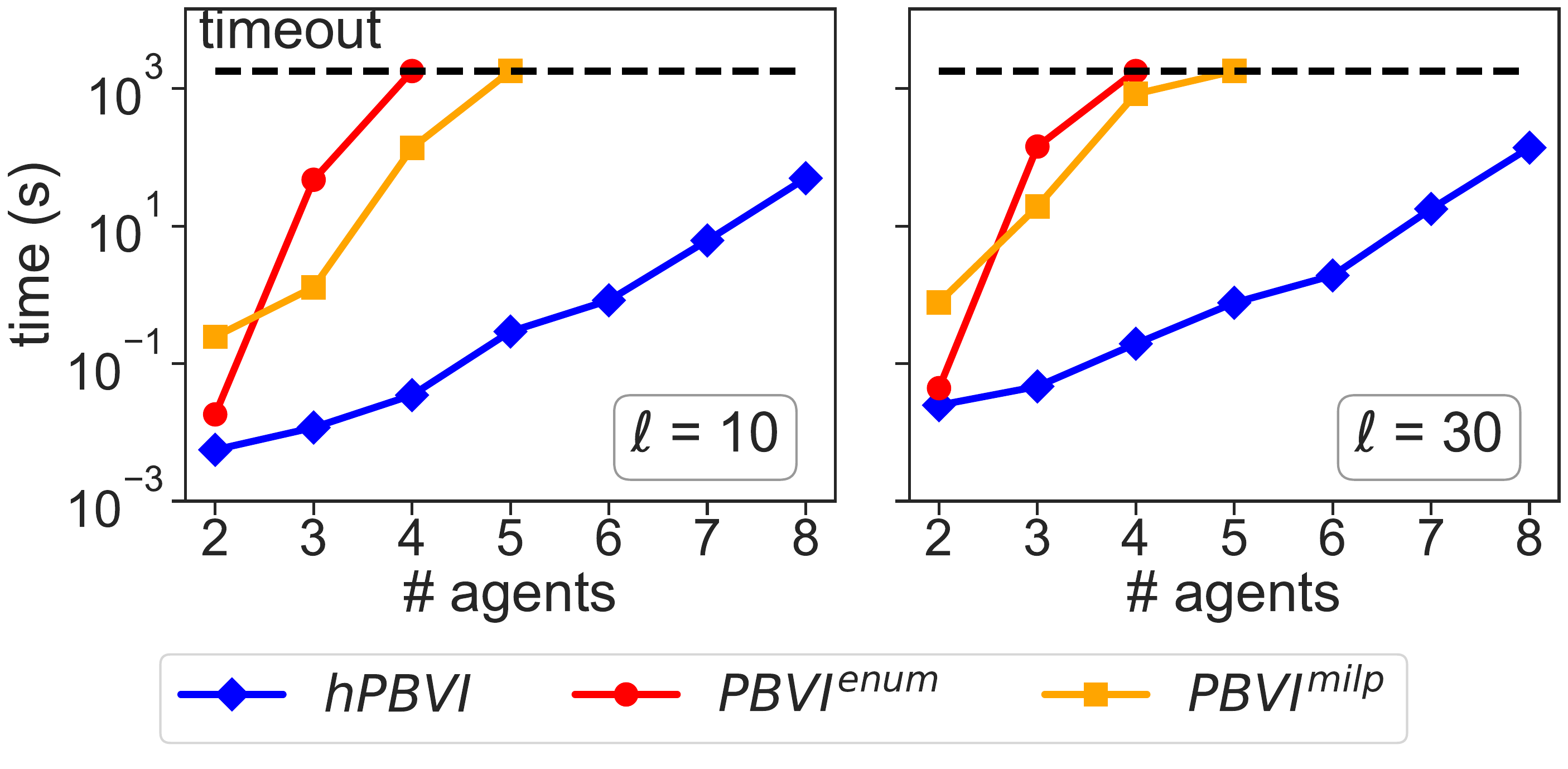}
    \vspace{-0.5cm}
    \caption{Average backup time for the recycling problem with different numbers of agents.}
\end{figure}

\begin{figure}[ht]
	\includegraphics[width=1\columnwidth]{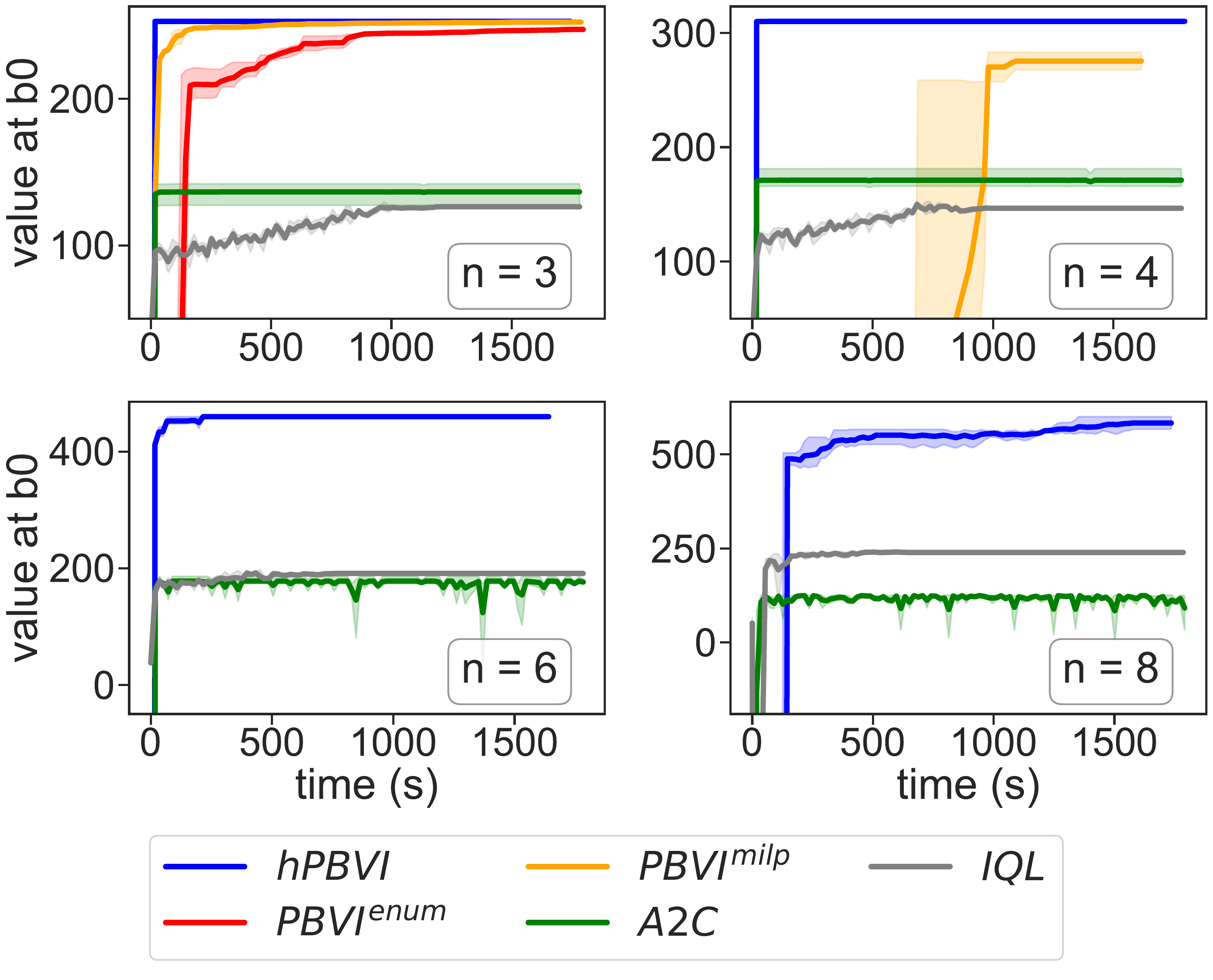}
    \vspace{-0.5cm}
	\caption{Anytime values for the recycling problem with teams of size $n\in\{3,4,6,8\}$ and planning  horizon $\ell = 30$.}
	\label{fig:anytimeCurves_recycling_main} 
\end{figure}

\section{Discussion}

This paper presents a point-based value iteration algorithm for near-optimally solving Dec-POMDPs. It exploits a hierarchical information-sharing structure, a dominant management style in our society for corporations, governments, criminal enterprises, armies, and religions. Under this assumption, it shows that point-based backup operations can be solved as perfect-information extensive-form games without compromising optimality. Doing so results in an exponential complexity drop, allowing global methods to scale up to larger teams of players. A thorough empirical analysis reveals that algorithms utilizing our findings can scale up to larger teams of players. In contrast, the state-of-the-art global approaches quickly ran out of resources. Another important empirical finding is that our approach scales to all medium-sized tested domains while providing equal or better performances than a state-of-the-art local method. 

Traditionally, global methods have been considered ineffective in games that involve medium to large-sized teams of players. For instance, state-of-the-art Dec-POMDP solvers such as FB-HSVI were only designed for two players \citep{OliehoekSDA10,DibangoyeMC09,DibangoyeABC13,Dibangoye2016}. However, we have presented a paper that puts forth several propositions for developing global methods that possess the scalability of local methods while maintaining global guarantees. In applications where the stakes are high and critical, such as search and rescue, security, and healthcare, scalable global methods with more reliable solutions than those from local methods are essential. 

Similarly to \citet{kovavrik2022rethinking}, our paper demonstrates that simultaneous-move games can be solved sequentially and centrally while allowing each player to act optimally in a decentralized manner.  This sequential and centralized training for decentralized execution (SCTDE) approach enables us to leverage private information available to players in a simple manner.  Additionally, the SCTDE approach enables us to reason for each player individually in a way that is similar to extensive-form games. This results in a significant reduction in complexity, especially when faced with public observations. This insight also allows us to transfer theories and algorithms from extensive-form games to simultaneous-move games. While our study demonstrates how to optimally solve single-stage games as extensive-form games, the principle we discussed also applies to planning and learning to act in multi-stage general-sum games. 

Our study focuses on analyzing the line hierarchical structure. Although previous studies, such as \citet{pmlr-v119-xie20a}, have successfully applied the SCTDE approach in two-player common-payoff games, our paper extends the research to multiple players. Additionally, our research can be further expanded to consider other structures, such as a tree structure, where players at the same level are independent. In the past, several forms of structure have been investigated such as dynamics independence \citep{BeckerZLG04}, weak-separability \citep{NairVTY05}, and delayed information-sharing \citep{nayyar2010optimal}. However, it is not clear how the hierarchical assumption affects these structures and the corresponding planning and learning theories. 

\section*{Acknowledgements}
This work was supported by ANR project Planning and Learning to Act in Systems of Multiple Agents under Grant ANR-19-CE23-0018, and ANR project Data and Prior, Machine Learning and Control under Grant ANR-19-CE23-0006, and ANR project Multi-Agent Trust Decision Process for the Internet of Things under Grant ANR-21-CE23-0016, all funded by French Agency ANR.

\section*{Impact Statement}

This paper presents work whose goal is to advance the field of Machine Learning. There are many potential societal consequences of our work, none which we feel must be specifically highlighted here.

 \bibliography{icml2024}
\bibliographystyle{icml2024}

\newpage
\appendix
\onecolumn

\section{The PBVI Algorithm}
\label{sec:pbvi:algorithm}

This section presents a pseudocode for the point-based value iteration algorithm to solve decentralized, partially observable Markov decision processes with hierarchical information sharing near-optimally.


 \begin{algorithm}[!ht]

${\mathtt{function}~ \mathtt{PBVI}()}$
\begin{algorithmic}
 	\STATE Initialize $\tilde{S}_{0:}$ and $\mathcal{V}_{0:}$.	
 	
	\WHILE{$\mathcal{V}_{0:}$ has not converged} 	
		\STATE $\mathtt{improve}(\mathcal{V}_{0:},\tilde{S}_{0:})$.
		\STATE $\tilde{S}_{0:} \gets \mathtt{expand}(\tilde{S}_{0:})$.
	 \ENDWHILE{}	
\end{algorithmic}
${\mathtt{function}~ \mathtt{improve}(\mathcal{V}_{0:},\tilde{S}_{0:})}$
\begin{algorithmic}
	\FOR{$\tau=\ell-1$ to $0$}
		\FOR{$s_\tau \in \tilde{S}_\tau$}
			\STATE $\mathcal{V}_\tau \gets \mathcal{V}_\tau\cup\{\mathtt{backup}(s_\tau, \mathcal{V}_{\tau+1})\}$. 
		\ENDFOR
	\ENDFOR
\end{algorithmic}
\caption{PBVI for $M'$ under HIS. }
\label{pbvi:his:dec:pomdp}
\end{algorithm}

\section{Proof of Theorem \ref{thm:sufficiency}}
\label{appendix:thm:sufficiency}

This section shows that the nested-occupancy state is a sufficient statistic for the central planner to optimally solve perfect-information extensive-form game $\bar{G}_{s_\tau}^{\beta_\tau}$ for any player $i$.

\begin{proof}
The nested-occupancy state is a sufficiency statistic of the total data available to the planner at any player $i$ for optimally solving $\bar{G}_{s_\tau}^{\beta_\tau}$, if it is sufficient to predict (1) the next nested-occupancy state and (2) the immediate reward.
Let the total data available to the planner at player $i$ be $\varsigma_\tau^i \doteq (s_\tau, o^i_\tau, u^{:i-1}_\tau)$.
Let $u^i_\tau$ be the action chosen at player $i$ after experiencing total data $\varsigma_\tau^i$.
The nested-occupancy state $s^i_\tau \doteq (b^i_\tau, o^i_\tau, u^{:i-1}_\tau)$ summarizing $\varsigma_\tau^i$ is sufficient to predict the next nested-occupancy state $s^{i+1}_\tau$,  if and only if the following holds  $\Pr\{s^{i+1}_\tau | \varsigma^i_\tau,u^i_\tau\} = \Pr\{s^{i+1}_\tau | s^i_\tau,u^i_\tau\}$.
To prove this property, we start with the definition of a nested-belief state $b^i_\tau$ associated with nested-occupancy state $s^i_\tau$, \ie
\begin{align*}
\Pr\{s^{i+1}_\tau | \varsigma^i_\tau, u_\tau^i\} &\doteq \Pr\{b^{i+1}_\tau, o^{i+1}_\tau, u^{:i,\bullet}_\tau | s_\tau, o^i_\tau, u^{:i,\circ}_\tau\},&\text{(by definition of $\varsigma_\tau^i$ and $s^i_\tau$)}\\
&= \Pr\{b^{i+1}_\tau, o^{i+1}_\tau | s_\tau, o^i_\tau,  u^{:i,\circ}_\tau\}  \cdot \Pr\{u^{:i,\bullet}_\tau | s_\tau, o^i_\tau,  u^{:i,\circ}_\tau\},&\text{(by application of the Bayes rule)}\\
&= \Pr\{b^{i+1}_\tau, o^{i+1}_\tau | \varsigma_\tau^i \}  \cdot \Pr\{u^{:i,\bullet}_\tau |   u^{:i,\circ}_\tau\},&\text{(by checking constraint $u^{:i,\bullet}_\tau = u^{:i,\circ}_\tau$)}\\
&= b^i_\tau(b^{i+1}_\tau, o^{i+1}_\tau) \cdot \mathbb{1}\{ u^{:i,\circ}_\tau = u^{:i,\bullet}_\tau \},&\text{(by definition of $b^i_\tau$)}.
\end{align*}
It will prove useful to define transition rule $\tilde{T}\colon (s^{i+1}_\tau|b^i_\tau, u_\tau^{:i}) \mapsto \Pr\{s^{i+1}_\tau | b^i_\tau, u_\tau^{:i}\}$, describing the probability to transitionning into nested-occupancy state $s^{i+1}_\tau$ upon talking action $u^i_\tau$ in nested-occupancy state $s^i_\tau$. Notice that the transition does not depend on the current history $o^i_\tau$, or it does so only through $(b^i_\tau, u_\tau^{:i})$.
Next, we show the sufficiency of nested-occupancy states to predict immediate rewards. Since rewards occur only at player $n$, we shall only consider nested occupancy states at that player. The nested-occupancy state $s^n_\tau \doteq (b^n_\tau, o^n_\tau, u^{:n-1}_\tau)$ summarizing $\varsigma_\tau^n$ is sufficient to predict the immediate reward upon taking action $u^n_\tau$,  if and only if there exists a reward function $(s_\tau^n,u^n_\tau) \mapsto \tilde{R} (s_\tau^n,u^n_\tau)$ such that the following holds: $R(\varsigma_\tau^n,u^n_\tau) =  \tilde{R} (s_\tau^n,u^n_\tau)$.  To prove this statement, we start with the definition of $R(\varsigma_\tau^n,u^n_\tau)$, \ie
\begin{align*}
R(\varsigma_\tau^n,u^n_\tau) &\doteq \mathbb{E}_{x\sim \Pr\{\cdot|\varsigma_\tau^n,u_\tau^n\}}\{\beta_\tau(x,o,u)\},& \text{(by definition of $R(\varsigma_\tau^n,u^n_\tau)$)}\\
&= \mathbb{E}_{x\sim b^n_\tau(\cdot)}\{\beta_\tau(x,o,u)\},& \text{(by definition of $b^n_\tau$)}.
\end{align*}
If we let $ \tilde{R} (s_\tau^n,u^n_\tau) \doteq \mathbb{E}_{x\sim b^n_\tau(\cdot)}\{\beta_\tau(x,o,u)\}$, then the statement holds.

Nested-occupancy states describe a perfect-information extensive-form game $\tilde{G}_{s_\tau}^{\beta_\tau}\doteq \langle n, \tilde{\Sigma}, \tilde{\Psi}, \tilde{T},\tilde{R} \rangle$ where nodes $\tilde{\Sigma}$ are nested-occupancy states,  $\tilde{\Psi} \colon \tilde{\Sigma} \to 2^{(\cup_{i=1}^n U^i)}$ specifies the action set available to each nested-occupancy state, and $\tilde{T}$ and $\tilde{R}$ are already defined. Clearly, any optimal solution for perfect-information extensive-form game $\tilde{G}_{s_\tau}^{\beta_\tau}$ is also optimal for perfect-information extensive-form game $\bar{G}_{s_\tau}^{\beta_\tau}$. To prove this statement, we need to prove that the optimal action-value functions $\tilde{\beta}_\tau^{1:n,*}$ of $\tilde{G}_{s_\tau}^{\beta_\tau}$ such that at any player $i$ and node $\varsigma_\tau^i$ (resp. nested-occupancy state $s^i_\tau$) and action $u^i_\tau$ the following equality holds: $\beta_\tau^{i,*}(\varsigma_\tau^i,u_\tau^i) = \tilde{\beta}_\tau^{i,*}(s_\tau^i,u_\tau^i)$. We prove the statement by induction. The statement trivially holds at player $n$ since $R(\varsigma_\tau^n,u^n_\tau) =  \tilde{R} (s_\tau^n,u^n_\tau)$, then  $\beta_\tau^{n,*}(\varsigma_\tau^i,u_\tau^i) \doteq  R(\varsigma_\tau^n,u^n_\tau) = \tilde{R} (s_\tau^n,u^n_\tau) \doteq \tilde{\beta}_\tau^{i,*}(s_\tau^i,u_\tau^i)$. Suppose the statement hold at player $i+1$ onward, \ie $\beta_\tau^{i+1,*}(\varsigma_\tau^{i+1},u_\tau^{i+1}) = \tilde{\beta}_\tau^{{i+1},*}(s_\tau^{i+1},u_\tau^{i+1})$. We are now ready to show it also hold at player $i$. We start with the expression of the optimal action-value function $\beta_\tau^{i,*}$, \ie
\begin{align*}
\beta^{i,*}_\tau(\varsigma_\tau^i, u^i_\tau) &= \textstyle \mathbb{E}_{\varsigma_\tau^{i+1}\sim T(\cdot|\varsigma_\tau^i, u^i_\tau) }\{  \max_{u^{i+1}_\tau}\beta_\tau^{i+1,*}(\varsigma_\tau^{i+1}, u_\tau^{i+1})\}, & \text{(by Theorem \ref{thm:bellman:equations})}\\
&=\textstyle  \mathbb{E}_{s_\tau^{i+1}\sim \Pr\{\cdot|\varsigma_\tau^i, u^i_\tau\} }\{  \max_{u^{i+1}_\tau}\tilde{\beta}_\tau^{i+1,*}(s_\tau^{i+1}, u_\tau^{i+1})\},& \text{(by induction hypothesis)}\\
&=\textstyle  \mathbb{E}_{s_\tau^{i+1}\sim \tilde{T}(\cdot|b_\tau^i, u^{:i}_\tau) }\{  \max_{u^{i+1}_\tau}\tilde{\beta}_\tau^{i+1,*}(s_\tau^{i+1}, u_\tau^{i+1})\},& \text{(by definition of $\tilde{T}(\cdot|b_\tau^i, u^{:i}_\tau) $)}\\
&= \tilde{\beta}^{i,*}_\tau(s_\tau^i, u^i_\tau) ,& \text{(by definition of $ \tilde{\beta}^{i,*}_\tau$)}.
\end{align*}
The statement holds for player $i$, thus for any arbitrary player. Consequently, one can use nested-occupancy states instead of total data available to the planner without compromising optimality, which ends the proof.
\end{proof}

\section{Equivalence Relations}

This section presents crucial properties necessary for grouping private histories that convey the same information about the game. To cluster two private histories of a player and reason similarly for the entire cluster, it is imperative to ensure that making identical immediate and future decisions for all members in the cluster does not compromise optimality. To do so, we need to specify how the information we rely on to make decisions evolves over stages. In particular, we need to exhibit rules for calculating the next-stage nested-occupancy state given the current one and decisions made at the current stage.

\subsection{Predicting Next-Stage Observations}

This subsection shows that the observation at stage $\tau+1$ and player $i$ can be accurately predicted using only the nested occupancy states at stage $\tau$ and player $i$ along with decision rules at stage $\tau$ and players $i$ to $n$.
\begin{lemma}
\label{lem:observation}
Let $\varsigma_\tau^i \doteq \langle s_\tau, o^i_\tau, u^{:i-1}_\tau \rangle$ be total data available to the planner at stage $\tau$ and player $i$. Let $a^{i:}_\tau$ be the decision rules for player $i$ to player $n$ at stage $\tau$. Let $s_\tau^i \doteq \langle b^i_\tau, o^i_\tau, u^{:i-1}_\tau \rangle$ be the nested-occupancy state at stage $\tau$ and player $i$ summarizing total data $\varsigma_\tau^i $. The probability $\Pr\{z^i_{\tau+1}| \varsigma^i_\tau, a^{i:}_\tau\}$ that the planner receives observation $z^i_{\tau+1}$ on behalf of player $i$ upon acting according to $\langle u^i_\tau, a^{i+1:}_\tau\rangle$ starting in total data $\varsigma_\tau^i$ satisfies the following recursion:
\begin{align}
\Omega^i(z^i_{\tau+1}| b^i_\tau,  u^{:i}_\tau, a^{i+1:}_\tau) &=  \sum_{s^{i+1}_\tau} \tilde{T}(s^{i+1}_\tau|b^i_\tau, u^{:i}_\tau) \sum_{z^{i+1}_{\tau+1}} \mathbb{1}\{ z^i_{\tau+1}\sqsubseteq \zeta^{i+1}( z^{i+1}_{\tau+1})\}\cdot \Omega^{i+1}(z^{i+1}_{\tau+1}| b^{i+1}_\tau,  u_\tau^{:i}, a^{i+1}_\tau(o^{i+1}_\tau), a^{i+2:}_\tau),
\label{eqn:lem:observation}
\end{align}
with boundary condition $\Omega^n(z^n_{\tau+1}| b^n_\tau,  u_\tau^{:n}) \doteq \sum_{x}\sum_{y} b^n_\tau(x)\cdot p(y,z^n_{\tau+1}|x,u_\tau^{:n})$.
\end{lemma}
\begin{proof}
Starting from conditional probability distribution $\Pr\{z^i_{\tau+1}| \varsigma^i_\tau,  u^i_\tau, a^{i+1:}_\tau\} $ and expanding over nested-occupancy states $s^{i+1}_\tau$ and observations  $z^{i+1}_{\tau+1}$ of player $i+1$
\begin{align*}
\Pr\{z^i_{\tau+1}| \varsigma^i_\tau,  u^i_\tau, a^{i+1:}_\tau\} 
&= \sum_{s^{i+1}_\tau}\sum_{z^{i+1}_{\tau+1}} \Pr\{s^{i+1}_\tau, z^{i+1}_{\tau+1}, z^i_{\tau+1}| \varsigma^i_\tau,  u^i_\tau, a^{i+1:}_\tau\}.
\end{align*}
The expansion of the joint probability into the product of conditional probabilities yields the following expression:
\begin{align}
&= \sum_{s^{i+1}_\tau}\sum_{z^{i+1}_{\tau+1}} \Pr\{ z^i_{\tau+1}| s^{i+1}_\tau, z^{i+1}_{\tau+1},\varsigma^i_\tau,u^i_\tau, a^{i+1:}_\tau\}\cdot \Pr\{z^{i+1}_{\tau+1}| s^{i+1}_\tau, \varsigma^i_\tau, u^i_\tau, a^{i+1:}_\tau\}\cdot \Pr\{s^{i+1}_\tau| \varsigma^i_\tau, u^i_\tau, a^{i+1:}_\tau\}.
\label{eqn:lem:observation:1}
\end{align}
The first factor in (\ref{eqn:lem:observation:1}) depends solely upon $z^{i+1}_{\tau+1}$ and not on the tuple $(s^{i+1}_\tau, \varsigma^i_\tau, u^i_\tau, a^{i+1:}_\tau)$, \ie  
\begin{align}
\Pr\{z^i_{\tau+1}| \varsigma^i_\tau, u^i_\tau, a^{i+1:}_\tau\} 
&= \sum_{s^{i+1}_\tau}\sum_{z^{i+1}_{\tau+1}} \Pr\{ z^i_{\tau+1}| z^{i+1}_{\tau+1}\}\cdot \Pr\{z^{i+1}_{\tau+1}| s^{i+1}_\tau, \varsigma^i_\tau, u^i_\tau, a^{i+1:}_\tau\}\cdot \Pr\{s^{i+1}_\tau| \varsigma^i_\tau, u^i_\tau, a^{i+1:}_\tau\}.
\label{eqn:lem:observation:02}
\end{align}
The last factor in (\ref{eqn:lem:observation:02}) depends solely upon $\langle b^i_\tau, u^{:i}_\tau\rangle$ and not on tuple $\langle \varsigma^i_\tau, a^{i+1:}_\tau \rangle$, which becomes after re-arranging terms:
\begin{align}
\Pr\{z^i_{\tau+1}| \varsigma^i_\tau, u^i_\tau, a^{i+1:}_\tau\} &= \sum_{s^{i+1}_\tau} \tilde{T}(s^{i+1}_\tau| b^i_\tau, u^{:i}_\tau)\sum_{z^{i+1}_{\tau+1}} \Pr\{ z^i_{\tau+1}| z^{i+1}_{\tau+1}\}\cdot \Pr\{z^{i+1}_{\tau+1}| s^{i+1}_\tau, \varsigma^i_\tau, u^i_\tau, a^{i+1:}_\tau\}.
\label{eqn:lem:observation:2}
\end{align}
Equation (\ref{eqn:lem:observation:2}) makes it possible to prove the statement, Equation (\ref{eqn:lem:observation}), recursively. We start at player $n-1$, \ie
\begin{align}
\Pr\{z^{n-1}_{\tau+1}| \varsigma^{n-1}_\tau, u^{n-1}_\tau, a^n_\tau\} &= \sum_{s^n_\tau} \tilde{T}(s^n_\tau|b^{n-1}_\tau,  u^{:n-1}_\tau)\sum_{z^n_{\tau+1}} \Pr\{ z^{n-1}_{\tau+1}| z^n_{\tau+1}\}\cdot \Pr\{z^n_{\tau+1}| s^n_\tau, \varsigma^{n-1}_\tau, u^{n-1}_\tau, a^n_\tau\}. \label{eqn:lem:observation:3}
\end{align}
The boundary condition gives the last factor in (\ref{eqn:lem:observation:3}), \ie for nested-occupancy state $s^n_\tau \doteq (b^n_\tau, o^n_\tau, u^{:n-1}_\tau)$,
\begin{align}
\Pr\{z^{n-1}_{\tau+1}| \varsigma^{n-1}_\tau, u^{n-1}_\tau, a^n_\tau\} &= \sum_{s^n_\tau} \tilde{T}(s^n_\tau|b^{n-1}_\tau,  u^{:n-1}_\tau)\sum_{z^n_{\tau+1}} \mathbb{1}\{ z^{n-1}_{\tau+1}\sqsubseteq z^n_{\tau+1}\}\cdot  \Omega^n(z^n_{\tau+1}| b^n_\tau, \langle u_\tau^{:n-1},a^n_\tau(o^n_\tau)\rangle) . \label{eqn:lem:observation:4}
\end{align}
Let $\Omega^{n-1}\colon (z^{n-1}_{\tau+1}| b^{n-1}_\tau,   u^{:n-1}_\tau, a^n_\tau) \mapsto  \ \sum_{s^n_\tau} \tilde{T}(s^n_\tau|b^{n-1}_\tau,  u^{:n-1}_\tau)\sum_{z^n_{\tau+1}} \mathbb{1}\{ z^{n-1}_{\tau+1}\sqsubseteq \zeta^n(z^n_{\tau+1})\}\cdot  \Omega^n(z^n_{\tau+1}| b^n_\tau,  u_\tau^{:n-1},a^n_\tau(o^n_\tau))$ be the observation model for predicting next observation at stage $\tau$ and player $n-1$. Then, statement (\ref{eqn:lem:observation}) holds at stage $\tau$ and player $n-1$. Suppose the statement holds for any arbitrary stage $\tau$ and player $i+1$. We are ready to prove it also holds at stage $\tau$ and player $i$. Starting at (\ref{eqn:lem:observation:2}),  the application of the induction hypothesis yields:
\begin{align}
\Pr\{z^i_{\tau+1}| \varsigma^i_\tau, u^i_\tau, a^{i+1:}_\tau\} &= \sum_{s^{i+1}_\tau} \tilde{T}(s^{i+1}_\tau| b^i_\tau, u^{:i}_\tau)\sum_{z^{i+1}_{\tau+1}} \mathbb{1}\{ z^i_{\tau+1}\sqsubseteq \zeta^{i+1}(z^{i+1}_{\tau+1})\}\cdot \Omega^{i+1}(z^{i+1}_{\tau+1}| b^{i+1}_\tau,  u^{:i}_\tau, a^{i+1:}_\tau).
\label{eqn:lem:observation:5}
\end{align}
If we let $\Omega^i\colon (z^i_{\tau+1}| b^i_\tau,  u^{:i}_\tau, a^{i+1:}_\tau) \mapsto \sum_{s^{i+1}_\tau} \tilde{T}(s^{i+1}_\tau| b^i_\tau, u^{:i}_\tau)\sum_{z^{i+1}_{\tau+1}} \mathbb{1}\{ z^i_{\tau+1}\sqsubseteq \zeta^{i+1}(z^{i+1}_{\tau+1})\}\cdot \Omega^{i+1}(z^{i+1}_{\tau+1}| b^{i+1}_\tau,  u^{:i}_\tau, a^{i+1:}_\tau)$ be the observation model for predicting next observation at stage $\tau$ and player $i$, then statement (\ref{eqn:lem:observation}) holds at stage $\tau$ and player $i$. Thus, the statement holds for any stage and player, which ends the proof.
\end{proof}

\subsection{Predicting Next-Stage Nested-Occupancy States}

This subsection proves the nested-occupancy states describe a Markovian process, \ie the next-stage nested-occupancy state depends only upon the current one. Notice that nested-occupancy states have three components. Only the nested belief states are nonobservable and need to be estimated. If we know how to estimate the nested belief state, we can add the history and actions of subordinates, thereby constructing a nested-occupancy state.

\begin{lemma}
\label{lem:markov}
Let $\varsigma_\tau^i \doteq \langle s_\tau, o^i_\tau, u^{:i-1}_\tau \rangle$ be total data available to the planner at stage $\tau$ and player $i$. Let $a^{i:}_\tau$ be the decision rules for player $i$ to player $n$ at stage $\tau$. Let $s_\tau^i \doteq \langle b^i_\tau,o^i_\tau, u^{:i-1}_\tau \rangle$ be the nested-occupancy state at stage $\tau$ and player $i$ summarizing total data $\varsigma_\tau^i $. The next-stage  nested-belief state $b^i_{\tau+1} \doteq T^i(b^i_\tau,\langle u^i_\tau, a^{i+1:}_\tau\rangle,z^i_{\tau+1})$, upon acting according to $\langle u^i_\tau, a^{i+1:}_\tau\rangle$ in nested-occupancy state $s_\tau^i $ and receiving observation $z^i_{\tau+1}$, satisfies the following recursion: for any history and nested-belief tuple $( o^{i+1}_{\tau+1}, b^{i+1}_{\tau+1})$,
\begin{align*}
b^i_{\tau+1}(o^{i+1}_{\tau+1}, b^{i+1}_{\tau+1})   
&\propto \textstyle \sum_{s^{i+1}_\tau\doteq (b^{i+1}_\tau,o^{i+1}_\tau,u_\tau^{:i})} \tilde{T}(s^{i+1}_\tau | b_\tau^i,  u^{:i}_\tau )\sum_{z^{i+1}_{\tau+1}}  
\mathbb{1}\{  z^i_{\tau+1}\sqsubseteq \zeta^{i+1}(z^{i+1}_{\tau+1})\} \cdot 
\delta_{\langle o^{i+1}_\tau, a^{i+1}_\tau(o^{i+1}_\tau),z^{i+1}_{\tau+1}\rangle}^{o^{i+1}_{\tau+1}}
 \nonumber\\
&\quad 
\delta_{T^{i+1}(b^{i+1}_\tau,\langle u^{:i}_\tau,a^{i+1}_\tau(o^{i+1}_\tau), a^{i+2:}_\tau\rangle ,z^{i+1}_{\tau+1})}^{b^{i+1}_{\tau+1}}
 \cdot 
\Omega^{i+1}(z^{i+1}_{\tau+1}| b^{i+1}_\tau,\langle u^{:i}_\tau,a^{i+1}_\tau(o^{i+1}_\tau), a^{i+2:}_\tau\rangle ).   
\end{align*}
with boundary condition $b^n_{\tau+1}\doteq  T^n(b^n_\tau, u^{:n}_\tau, z^n_{\tau+1})   $ where  $ b^n_{\tau+1}(y) \propto \sum_{x} b^n_\tau(x)\cdot p(y,z^n_{\tau+1}|x,u^{:n}_\tau)$ for any hidden state $y$. 
\end{lemma}
\begin{proof}
The proof proceeds by induction. Starting with player $n-1$, we define the nested-belief state $b^{n-1}_{\tau+1}$ at stage $\tau+1$ and player $n-1$ upon acting according to $\langle u^{n-1}_\tau, a^n_\tau\rangle$ in total data available to the planner $\varsigma_\tau^{n-1} \doteq \langle s_\tau, o^{n-1}_\tau, u^{:n-2}_\tau \rangle$ and receiving observation $z^{n-1}_{\tau+1}$, as follows: for any history and nested-belief tuple $( o^n_{\tau+1}, b^n_{\tau+1})$,
\begin{align}
b^{n-1}_{\tau+1}( o^n_{\tau+1}, b^n_{\tau+1}) &\doteq \Pr\{ o^n_{\tau+1}, b^n_{\tau+1} |\varsigma_\tau^{n-1},  u^{n-1}_\tau, a^n_\tau, z^{n-1}_{\tau+1}\}.   
\label{eqn:lem:markov:1}
\end{align}
The expansion of (\ref{eqn:lem:markov:1}) over nested-occupancy states $s^n_\tau$ at stage $\tau$ and player $n$ and histories $z^n_{\tau+1}$ at stage $\tau+1$ and player $n$, result in the following expression:  
\begin{align}
b^{n-1}_{\tau+1}(o^n_{\tau+1}, b^n_{\tau+1})  &=
\sum_{s^n_\tau\doteq (b^n_\tau,o^n_\tau,u_\tau^{:n-1})}\sum_{z^n_{\tau+1}} 
\Pr\{s^n_\tau,z^n_{\tau+1},o^n_{\tau+1}, b^n_{\tau+1} |\varsigma_\tau^{n-1},  u^{n-1}_\tau, a^n_\tau, z^{n-1}_{\tau+1}\}.   
\label{eqn:lem:markov:2}
\end{align}
The application of Bayes' rule in expression (\ref{eqn:lem:markov:2}) yields the following expression:
\begin{align}
b^{n-1}_{\tau+1}(o^n_{\tau+1}, b^n_{\tau+1})   &=
\sum_{s^n_\tau\doteq (b^n_\tau,o^n_\tau,u_\tau^{:n-1})}\sum_{z^n_{\tau+1}} 
\frac{\Pr\{s^n_\tau,z^n_{\tau+1},o^n_{\tau+1}, b^n_{\tau+1} ,\varsigma_\tau^{n-1},  u^{n-1}_\tau, a^n_\tau, z^{n-1}_{\tau+1}\}
}{\Pr\{\varsigma_\tau^{n-1},  u^{n-1}_\tau, a^n_\tau, z^{n-1}_{\tau+1}\}}.   
\label{eqn:lem:markov:3}
\end{align}
The expansion of the joint probability into the product of conditional probabilities and the application of Lemma \ref{lem:observation} yield the following expression: 
\begin{align}
b^{n-1}_{\tau+1}(o^n_{\tau+1}, b^n_{\tau+1})   &=
\textstyle \sum_{s^n_\tau\doteq (b^n_\tau,o^n_\tau,u_\tau^{:n-1})}\sum_{z^n_{\tau+1}} 
\Pr\{z^{n-1}_{\tau+1}| z^n_{\tau+1}\} \cdot 
\Pr\{o^n_{\tau+1}| o^n_\tau, a^n_\tau(o^n_\tau),z^n_{\tau+1}\} \cdot 
\Pr\{b^n_{\tau+1}|b^n_\tau,\langle u^{:n-1}_\tau,a^n_\tau(o^n_\tau)\rangle ,z^n_{\tau+1}\} \cdot \nonumber\\
&\quad 
\Pr\{z^n_{\tau+1}| b^n_\tau,\langle u^{:n-1}_\tau,a^n_\tau(o^n_\tau)\rangle \} \cdot
\Pr\{s^n_\tau | \varsigma_\tau^{n-1},  u^{n-1}_\tau \}
/ \Pr\{z^{n-1}_{\tau+1}|\varsigma_\tau^{n-1},  u^{n-1}_\tau, a^n_\tau\}
\label{eqn:lem:markov:4}
\end{align}
Using the boundary condition, we obtain the following expression, \ie 
\begin{align}
b^{n-1}_{\tau+1}(o^n_{\tau+1}, b^n_{\tau+1})   &\propto
\textstyle \sum_{s^n_\tau\doteq (b^n_\tau,o^n_\tau,u_\tau^{:n-1})}\sum_{z^n_{\tau+1}} 
 \mathbb{1}\{  z^{n-1}_{\tau+1}\sqsubseteq \zeta^n(z^n_{\tau+1})\} \cdot 
\delta_{\langle o^n_\tau, a^n_\tau(o^n_\tau),z^n_{\tau+1}\rangle}^{o^n_{\tau+1}}
 \nonumber\\
&\quad 
\delta_{T^n(b^n_\tau,\langle u^{:n-1}_\tau,a^n_\tau(o^n_\tau)\rangle ,z^n_{\tau+1})}^{b^n_{\tau+1}}
 \cdot 
\Omega^n(z^n_{\tau+1}| b^n_\tau,\langle u^{:n-1}_\tau,a^n_\tau(o^n_\tau)\rangle )\cdot
\tilde{T}(s^n_\tau | b_\tau^{n-1},  u^{:n-1}_\tau )
\label{eqn:lem:markov:5}
\end{align}
Hence, the statement holds at stage $\tau$ and player $n-1$. Suppose it holds at stage $\tau$ and player $i+1$. We are now ready to show the statement also holds at stage $\tau$ and player $i$. We start with the definition the nested-belief state $b^i_{\tau+1}$ at stage $\tau+1$ and player $i$ upon acting according to $\langle u^i_\tau, a^{i+1:}_\tau\rangle$ in total data available to the planner $\varsigma_\tau^i \doteq \langle s_\tau, o^i_\tau, u^{:i-1}_\tau \rangle$ and receiving observation $z^i_{\tau+1}$. The proof proceeds similarly to that of player $n-1$, \ie
\begin{align}
b^i_{\tau+1}(o^{i+1}_{\tau+1}, b^{i+1}_{\tau+1}) &\doteq \Pr\{o^{i+1}_{\tau+1}, b^{i+1}_{\tau+1} |\varsigma_\tau^i,  u^i_\tau, a^{i+1:}_\tau, z^i_{\tau+1}\}.   
\label{eqn:lem:markov:6}
\end{align}
The expansion of (\ref{eqn:lem:markov:6}) over nested-occupancy states $s^{i+1}_\tau$ at stage $\tau$ and player $i+1$ and histories $z^{i+1}_{\tau+1}$ at stage $\tau+1$ and player $i+1$, result in the following expression: 
\begin{align}
b^i_{\tau+1}(o^{i+1}_{\tau+1}, b^{i+1}_{\tau+1}) &=\textstyle\sum_{s^{i+1}_\tau\doteq (b^{i+1}_\tau,o^{i+1}_\tau,u_\tau^{:i})}\sum_{z^{i+1}_{\tau+1}} \Pr\{s^{i+1}_\tau,z^{i+1}_{\tau+1},o^{i+1}_{\tau+1}, b^{i+1}_{\tau+1} |\varsigma_\tau^i,  u^i_\tau, a^{i+1:}_\tau, z^i_{\tau+1}\}.   
\label{eqn:lem:markov:7}
\end{align}
The application of Bayes' rule in expression (\ref{eqn:lem:markov:7}) yields the following expression: for any pairs $(o^{1:i+1}_{\tau+1}, b^{i+1}_{\tau+1})$, 
\begin{align}
b^i_{\tau+1}(o^{i+1}_{\tau+1}, b^{i+1}_{\tau+1})  &= \sum_{s^{i+1}_\tau\doteq (b^{i+1}_\tau,o^{i+1}_\tau,u_\tau^{:i})}\sum_{z^{i+1}_{\tau+1}}  \Pr\{s^{i+1}_\tau,z^{i+1}_{\tau+1},o^{i+1}_{\tau+1}, b^{i+1}_{\tau+1} ,\varsigma_\tau^i,  u^i_\tau, a^{i+1:}_\tau, z^i_{\tau+1}\} / \Pr\{\varsigma_\tau^i,  u^i_\tau, a^{i+1:}_\tau, z^i_{\tau+1}\}.   
\label{eqn:lem:markov:8}
\end{align}
The expansion of the joint probability into the product of conditional probabilities and the application of Lemma \ref{lem:observation} yield the following expression: 
\begin{align}
b^i_{\tau+1}(o^{i+1}_{\tau+1}, b^{i+1}_{\tau+1})   
&=
\textstyle \sum_{s^{i+1}_\tau\doteq (b^{i+1}_\tau,o^{i+1}_\tau,u_\tau^{:i})}\sum_{z^{i+1}_{\tau+1}} 
\Pr\{z^i_{\tau+1}| z^{i+1}_{\tau+1}\} \cdot 
\Pr\{o^{i+1}_{\tau+1}| o^{i+1}_\tau, a^{i+1}_\tau(o^{i+1}_\tau),z^{i+1}_{\tau+1}\} \cdot \nonumber\\
&\quad 
\Pr\{b^{i+1}_{\tau+1}|b^{i+1}_\tau,\langle u^{:i}_\tau,a^{i+1}_\tau(o^{i+1}_\tau), a^{i+2:}_\tau\rangle ,z^{i+1}_{\tau+1}\} \cdot 
\Pr\{z^{i+1}_{\tau+1}| b^{i+1}_\tau,\langle u^{:i}_\tau,a^{i+1}_\tau(o^{i+1}_\tau), a^{i+2:}_\tau\rangle \} \cdot \nonumber\\
&\quad 
\Pr\{s^{i+1}_\tau | \varsigma_\tau^i,  u^i_\tau \}
/ \Pr\{z^i_{\tau+1}|\varsigma_\tau^i,  u^i_\tau, a^{i+1:}_\tau\}
\label{eqn:lem:markov:4}
\end{align}
Using the boundary condition and the induction hypothesis, we obtain the following expression, \ie 
\begin{align}
b^i_{\tau+1}(o^{i+1}_{\tau+1}, b^{i+1}_{\tau+1})   
&\propto \textstyle \sum_{s^{i+1}_\tau\doteq (b^{i+1}_\tau,o^{i+1}_\tau,u_\tau^{:i})} \tilde{T}(s^{i+1}_\tau | b_\tau^i,  u^{:i}_\tau )\sum_{z^{i+1}_{\tau+1}}  
\mathbb{1}\{  z^i_{\tau+1}\sqsubseteq \zeta^{i+1}(z^{i+1}_{\tau+1})\} \cdot 
\delta_{\langle o^{i+1}_\tau, a^{i+1}_\tau(o^{i+1}_\tau),z^{i+1}_{\tau+1}\rangle}^{o^{i+1}_{\tau+1}}
 \nonumber\\
&\quad 
\delta_{T^{i+1}(b^{i+1}_\tau,\langle u^{:i}_\tau,a^{i+1}_\tau(o^{i+1}_\tau), a^{i+2:}_\tau\rangle ,z^{i+1}_{\tau+1})}^{b^{i+1}_{\tau+1}}
 \cdot 
\Omega^{i+1}(z^{i+1}_{\tau+1}| b^{i+1}_\tau,\langle u^{:i}_\tau,a^{i+1}_\tau(o^{i+1}_\tau), a^{i+2:}_\tau\rangle ).   
\label{eqn:lem:markov:9}
\end{align}
Hence, the statement holds at any stage $\tau$ and player $i$, which ends the proof.
\end{proof}

\subsection{Nested Belief States, Policies and Action-Value Functions At Player $n$}

This section establishes many important properties regarding player $n$. First, it establishes that for any given stage and player $n$, the planner can make decisions based on belief states instead of histories. To prove this statement, one must demonstrate that belief states are capable of predicting (1) the next observation for the subsequent stage and player $n$; (2) the next belief state for the subsequent stage and player $n$; and (3) the immediate reward. Next, it shows that belief-dependent policies are optimal at player $n$. Finally, it describes the action-value functions under a history-dependent policy of player $1$ to $n-1$ and a belief-dependent policy of player $n$.

\begin{lemma}
\label{lem:sufficiency:n}
Let $a^{:n-1}_\tau$ be a joint decision rule of player $1$ to $n-1$. Let $b^n_\tau$ be a nested belief state, $o^{:n-1}_\tau$ be a joint history  of player $1$ to $n-1$, and $a^{:n-1}_\tau(o^{:n-1}_\tau)$ be a joint action of player $1$ to $n-1$ when following joint policy $a^{:n-1}_\tau$. Let $s^n_\tau \doteq (b^n_\tau, o^n_\tau, a^{:n-1}_\tau(o^{:n-1}_\tau))$ be a nested-occupancy state at stage $\tau$ and player $n$.
Then, the following propositions hold for any nested-occupancy state $s^n_\tau$ at stage $\tau$ and player $n$.
\begin{enumerate}
\item For any action $u^n_\tau$ and observation $z^n_{\tau+1}$, we have $\Pr\{z^n_{\tau+1} | s^n_\tau, u^n_\tau\} = \Omega^n(z^n_{\tau+1}|b_\tau^n,\langle u_\tau^n, a^{:n-1}_\tau(o^{:n-1}_\tau)\rangle)$.
\item For any action $u^n_\tau$ and observation $z^n_{\tau+1}$, we have $\Pr\{b^n_{\tau+1} | s^n_\tau, u^n_\tau, z^n_{\tau+1}\} = \delta^{b^n_{\tau+1}}_{T^n(b_\tau^n,\langle u_\tau^n, a^{:n-1}_\tau(o^{:n-1}_\tau)\rangle, z^n_{\tau+1})}$.
\item For any action $u^n_\tau$, we have $\mathbb{E}_{(x_\tau,u_\tau)\sim \Pr\{\cdot | s^n_\tau, u^n_\tau\} }\{ r(x_\tau,u_\tau) \} = \mathbb{E}_{(x_\tau,u_\tau)\sim \Pr\{\cdot | b^n_\tau, u_\tau^n, a^{:n-1}_\tau(o^{:n-1}_\tau)\} }\{ r(x_\tau,u_\tau) \}$.
\end{enumerate}
\end{lemma}
\begin{proof}
The two first propositions hold directly from Lemmas \ref{lem:observation} and \ref{lem:markov}. The last proposition holds because the immediate rewards depend on the histories of player $n$ only through the corresponding belief states. Which ends the proof.
\end{proof}

\begin{lemma}
\label{lem:optimal:belief:dependent:policy}
The optimal policy of player $n$ depends only upon the belief state not on histories.
\end{lemma}
\begin{proof}
The proof proceeds by induction. Let $a^{:n-1}_{0:}$ be the joint policy of player $1$ to $n-1$. The best-response decision rule of player $n$ at stage $\ell-1$ is written as follows: for any history $o^n_{\ell-1}$,
\begin{align*}
a^n_{\ell-1}(o^n_{\ell-1}) &\textstyle \in \argmax_{u^n_{\ell-1}}~\mathbb{E}_{(x_{\ell-1},u_{\ell-1})\sim \Pr\{\cdot | b^n_{\ell-1}, o^n_{\ell-1}, u^n_{\ell-1}, a^{:n-1}_{\ell-1}(o^{:n-1}_{\ell-1})\} }\{ r(x_{\ell-1},u_{\ell-1}) \}&\text{(by Definition)}\\
&\textstyle \in \argmax_{u^n_{\ell-1}}~\mathbb{E}_{(x_{\ell-1},u_{\ell-1})\sim \Pr\{\cdot | b^n_{\ell-1}, u_{\ell-1}^n, a^{:n-1}_{\ell-1}(o^{:n-1}_{\ell-1})\} }\{ r(x_{\ell-1},u_{\ell-1}) \}.&\text{(by Lemma \ref{lem:sufficiency:n})}
\end{align*}
The statement holds at stage $\ell-1$. Define the value function $\tilde{\alpha}^n_{\ell-1}$ under the joint policy $a^{:n-1}_{0:}$ of player $1$ to $n-1$, $$\bar{\alpha}^n_{\ell-1} \colon (b^n_{\ell-1}, o^{:n-1}_{\ell-1}) \mapsto\textstyle \max_{u^n_{\ell-1}}~ \mathbb{E}_{(x_{\ell-1},u_{\ell-1})\sim \Pr\{\cdot | b^n_{\ell-1},u_{\ell-1}^n, a^{:n-1}_{\ell-1}(o^{:n-1}_{\ell-1})\} }\{ r(x_{\ell-1},u_{\ell-1}) \}.$$ Define  the value function $\bar{\beta}^n_{\ell-2}$ under the joint policy $a^{:n-1}_{0:}$ of player $1$ to $n-1$, $$\bar{\beta}^n_{\ell-2} \colon (b^n_{\ell-2}, o^{:n-1}_{\ell-2}, u_{\ell-2}) \mapsto\textstyle  \mathbb{E}_{(x_{\ell-2},b^n_{\ell-1},o^{:n-1}_{\ell-1})\sim \Pr\{\cdot | b^n_{\ell-2},u_{\ell-2},o^{:n-1}_{\ell-2}\} }\{ r(x_{\ell-2},u_{\ell-2}) + \gamma \bar{\alpha}^n_{\ell-1}(b^n_{\ell-1}, o^{:n-1}_{\ell-1})\}.$$
The best-response decision rule of player $n$ at stage $\ell-2$ is written as follows: for any history $o^n_{\ell-2}$,
\begin{align*}
a^n_{\ell-2}(o^n_{\ell-2}) &\textstyle \in \argmax_{u^n_{\ell-2}}~\bar{\beta}^n_{\ell-2}(b^n_{\ell-2}, o^{:n-1}_{\ell-2}, \langle u_{\ell-2}^n, a^{:n-1}_{\ell-2}(o^{:n-1}_{\ell-2})\rangle).
\end{align*}
Consequently, the statement holds for stages $\ell-1$ and $\ell-2$.  Suppose the statement holds for stage $\tau+1$, that is there exists an action-value function $\bar{\beta}^n_\tau$ under the joint policy $a^{:n-1}_{0:}$ of player $1$ to $n-1$, $$\bar{\beta}^n_\tau\colon (b^n_\tau, o^{:n-1}_\tau, u_\tau) \mapsto\textstyle  \mathbb{E}_{(x_\tau,b^n_{\tau+1},o^{:n-1}_{\tau+1})\sim \Pr\{\cdot | b^n_\tau,u_\tau,o^{:n-1}_\tau\} }\{ r(x_\tau,u_\tau) + \gamma \bar{\alpha}^n_{\tau+1}(b^n_{\tau+1}, o^{:n-1}_{\tau+1})\}.$$
We are now ready to show the statement also holds at stage $\tau$.
The best-response decision rule of player $n$ at stage $\tau$ is written as follows: for any history $o^n_\tau$,
\begin{align*}
a^n_\tau(o^n_\tau) &\textstyle \in \argmax_{u^n_\tau}~\bar{\beta}^n_\tau(b^n_\tau, o^{:n-1}_\tau, \langle u_\tau^n, a^{:n-1}_\tau(o^{:n-1}_\tau)\rangle).
\end{align*}
This proves the statement holds for stage $\tau$, ending the proof.
\end{proof}

Lemma \ref{lem:optimal:belief:dependent:policy} shows that the HIS assumption allows player $n$ to act based solely upon belief states instead of histories optimally. In other words, a belief-dependent policy exists as good or better than any history-dependent policy of player $n$. The subsequent lemma shows how the use of belief-dependent policies for player $n$ affects the description of the action-value function under a joint history-dependent policy of player $1$ to $n-1$ and a belief-dependent policy for player $n$.

\begin{lemma}
\label{lem:action:value:fct:n}
Let $a^{:n-1}_{0:}$ be the joint history-dependent policy of player $1$ to $n-1$, $\tilde{a}^n_{0:}$ be the belief-dependent policy of player $n$. The action-value function under joint policy $\langle a^{:n-1}_{0:}, \tilde{a}^n_{0:} \rangle$ is given as follows: 
$$\bar{\beta}^n_\tau\colon(b^n_\tau, o^{:n-1}_\tau, u^{:n}_\tau) \mapsto   \mathbb{E}_{(x_\tau,b^n_{\tau+1},o^{:n-1}_{\tau+1})\sim \Pr\{\cdot | b^n_\tau,u_\tau,o^{:n-1}_\tau\} }\{ r(x_\tau,u_\tau) + \gamma \bar{\beta}^n_{\tau+1}(b^n_{\tau+1}, o^{:n-1}_{\tau+1}, \langle a^{:n-1}_{\tau+1}(o^{:n-1}_{\tau+1}), \tilde{a}^n_{\tau+1}(b^n_{\tau+1}) \rangle)\}$$ with boundary condition $\bar{\beta}^n_\ell(\cdot,\cdot,\cdot) \doteq 0$.
\end{lemma}
\begin{proof}
The proof follows directly from the proof of Lemma \ref{lem:optimal:belief:dependent:policy}.
\end{proof}

\subsection{Proof of Theorem \ref{thm:compression}}
\label{appendix:thm:compression}

\begin{proof}
We shall treat each proposition separately. 

\paragraph{Statement 1.} 
 We prove the first statement by induction. We begin the proof by demonstrating that the statement holds at player $n$. Let $s^{n,\circ}_\tau \doteq (b^{n,\circ}_\tau,o^{n,\circ}_\tau,u^{:n-1,\circ}_\tau)$ and $s^{n,\textcolor{gray}{\bullet}}_\tau \doteq (b^{n,\textcolor{gray}{\bullet}}_\tau,o^{n,\textcolor{gray}{\bullet}}_\tau,u^{:n-1,\textcolor{gray}{\bullet}}_\tau)$ be two nested-occupancy states. Suppose $s^{n,\circ}_\tau \sim_{\mathscr{B}_1} s^{n,\textcolor{gray}{\bullet}}_\tau $, that is $(b^{n,\circ}_\tau,u^{:n-1,\circ}_\tau) = (b^{n,\textcolor{gray}{\bullet}}_\tau,u^{:n-1,\textcolor{gray}{\bullet}}_\tau)$. Consider the action-value function $\tilde{\beta}^{n,*}_\tau$ at stage $\tau$, player $n$, nested-occupancy state $s^{n,\circ}_\tau$ and action $u^n_\tau$, \ie
\begin{align*}
\tilde{\beta}^{n,*}_\tau(s^{n,\circ}_\tau, u^n_\tau) &\doteq \mathbb{E}_{x\sim b^{n,\circ}_\tau(\cdot)}\{\beta_\tau(x,o^{n,\circ}_\tau,\langle u^{:n-1,\circ}_\tau, u^n_\tau\rangle)\}\\
&=\bar{\beta}^n_\tau(b^{n,\circ}_\tau, o^{:n-1,\circ}_\tau,\langle u^{:n-1,\circ}_\tau, u^n_\tau\rangle)&\text{(by Lemma \ref{lem:action:value:fct:n})}.
\end{align*}
Since the action-value function $\bar{\beta}^n_\tau$ depends on the nested-occupancy state only through the belief state $b^{n,\circ}_\tau$, joint history $o^{:n-1,\circ}_\tau$, and joint action $\langle u^{:n-1,\circ}_\tau, u^n_\tau\rangle$, not upon joint history $o^{n,\circ}_\tau$, thus does the action-value function $\tilde{\beta}^{n,*}_\tau$. Hence, the first statement holds at player $n$. Suppose the statement holds for player $i+1$. We are now ready to show it also holds at player $i$.  We start with the expression of  optimal action-value  $\tilde{\beta}_\tau^{i,*}(s^{i,\circ}_\tau, u^i_\tau)$ for nested-occupancy state $s^{i,\circ}_\tau$ and action $u^i_\tau$, \ie
\begin{align*}
\tilde{\beta}_\tau^{i,*}(s^{i,\circ}_\tau, u^i_\tau) &= \textstyle  \mathbb{E}_{s_\tau^{i+1,\circ}\sim \tilde{T}(\cdot|b_\tau^{i,\circ},\langle u^i_\tau, u^{:i-1,\circ}_\tau \rangle) }\{  \max_{u^{i+1}_\tau}\tilde{\beta}_\tau^{i+1,*}(s_\tau^{i+1,\circ}, u_\tau^{i+1})\}.
\end{align*}
An inspection of the transition function $ \tilde{T}(\cdot|b_\tau^{i,\circ}, u^{:i}_\tau)$ reveals that it depends on nested-occupancy state $s_\tau^{i,\circ}$ only though nested-belief state $b_\tau^{i,\circ}$ and joint action $u^{:i}_\tau$. Consequently, if we let $s^{i,\circ}_\tau \sim_{\mathscr{B}_1} s^{i,\textcolor{gray}{\bullet}}_\tau$ then we know $(b^{i,\circ}_\tau,u^{:i, \circ}_\tau) = (b^{i,\textcolor{gray}{\bullet}}_\tau,u^{:i, \textcolor{gray}{\bullet}}_\tau)$, which leads to the statement:
\begin{align*}
\tilde{\beta}_\tau^{i,*}(s^{i,\circ}_\tau, u^i_\tau) &= \textstyle  \mathbb{E}_{s_\tau^{i+1,\textcolor{gray}{\bullet}}\sim \tilde{T}(\cdot|b_\tau^{i,\textcolor{gray}{\bullet}}, \langle u^i_\tau, u^{:i-1,\textcolor{gray}{\bullet}}_\tau \rangle) }\{  \max_{u^{i+1}_\tau}\tilde{\beta}_\tau^{i+1,*}(s_\tau^{i+1,\textcolor{gray}{\bullet}}, u_\tau^{i+1})\}\\
&= \tilde{\beta}_\tau^{i,*}(s^{i,\textcolor{gray}{\bullet}}_\tau, u^i_\tau).
\end{align*}
This expression proves the first statement at any stage $\tau$ and player $i$.
\paragraph{Statement 2.} 
  To prove the second statement, we build upon the first statement. If we let $s^{i,\circ}_\tau \sim_{\mathscr{B}_2} s^{i,\textcolor{gray}{\bullet}}_\tau$, then  for any arbitrary joint action $u_\tau^{:i-1}$ we know that $(b^{i,\circ}_\tau, o^{i,\circ}_\tau, u_\tau^{:i-1}) \sim_{\mathscr{B}_1} (b^{i,\textcolor{gray}{\bullet}}_\tau, o^{i,\textcolor{gray}{\bullet}}_\tau, u_\tau^{:i-1})$. If $s^{i,\circ}_\tau \sim_{\mathscr{B}_2} s^{i,\textcolor{gray}{\bullet}}_\tau$, we know that histories of subordinates are identical $o^{:i-1,\circ}_\tau=o^{:i-1,\textcolor{gray}{\bullet}}_\tau$, then the following holds $u_\tau^{:i-1} = a_\tau^{:i-1,*}(o^{:i-1,\textcolor{gray}{\bullet}}_\tau) = a_\tau^{:i-1,*}(o^{:i-1,\circ}_\tau) $. Consequently, by the application of the first statement, we have for any arbitrary action $u^i_\tau$, 
$$\tilde{\beta}_\tau^{i,*}(\langle  b^{i,\circ}_\tau, o^{i,\circ}_\tau, a_\tau^{:i-1,*}(o^{:i-1,\circ}_\tau)\rangle, u^i_\tau) = \tilde{\beta}_\tau^{i,*}(\langle b^{i,\textcolor{gray}{\bullet}}_\tau, o^{i,\textcolor{gray}{\bullet}}_\tau,  a_\tau^{:i-1,*}(o^{:i-1,\textcolor{gray}{\bullet}}_\tau)\rangle, u^i_\tau).$$ 
Consequently, the sets of optimal actions $A^{i,*}_\tau(o^{i,\circ}_\tau) \doteq \argmax_{u^i_\tau} \tilde{\beta}_\tau^{i,*}(\langle b^{i,\circ}_\tau, o^{i,\circ}_\tau, a_\tau^{:i-1,*}(o^{:i-1,\circ}_\tau)\rangle, u^i_\tau)$ and $A^{i,*}_\tau(o^{i,\textcolor{gray}{\bullet}}_\tau) \doteq \argmax_{u^i_\tau} \tilde{\beta}_\tau^{i,*}(\langle b^{i,\textcolor{gray}{\bullet}}_\tau, o^{i,\textcolor{gray}{\bullet}}_\tau, a_\tau^{:i-1,*}(o^{:i-1,\textcolor{gray}{\bullet}}_\tau)\rangle, u^i_\tau)$ at histories $ o^{i,\circ}_\tau$ and $o^{i,\textcolor{gray}{\bullet}}_\tau$, respectively, are equivalent, \ie $A^{i,*}_\tau(o^{i,\circ}_\tau)  = A^{i,*}_\tau(o^{i,\textcolor{gray}{\bullet}}_\tau) $.
 Since $a^{i,*}_\tau(o^{i,\textcolor{gray}{\bullet}}_\tau)$ and $a^{i,*}_\tau(o^{i,\circ}_\tau)$ belong to the same set $A^{i,*}_\tau(o^{i,\circ}_\tau)  = A^{i,*}_\tau(o^{i,\textcolor{gray}{\bullet}}_\tau) $, they are  interchangeable.  In other words, the optimal action for history $o^{i,\circ}_\tau$ is also optimal for history $o^{i,\textcolor{gray}{\bullet}}_\tau$ and vice versa.  Interestingly, one can show that the expansions $\langle o^{i,\circ}_\tau, u^i_\tau, z^i_{\tau+1}\rangle$ and $\langle o^{i,\textcolor{gray}{\bullet}}_\tau, u^i_\tau, z^i_{\tau+1}\rangle$ of histories $o^{i,\circ}_\tau$ and $o^{i,\textcolor{gray}{\bullet}}_\tau$ upon taking the same action $u^i_\tau$ and receiving the same observation $z^i_{\tau+1}$, respectively, will also have equivalent optimal actions. Hence, essentially providing that the optimal policy for history $o^{i,\circ}_\tau$ is also optimal for history $o^{i,\textcolor{gray}{\bullet}}_\tau$ and vice versa. To show this statement, first notice that both histories $\langle o^{i,\circ}_\tau, u^i_\tau, z^i_{\tau+1}\rangle$ and $\langle o^{i,\textcolor{gray}{\bullet}}_\tau, u^i_\tau, z^i_{\tau+1}\rangle$ will have the same histories of subordinates because original histories $o^{i,\circ}_\tau$ and $o^{i,\textcolor{gray}{\bullet}}_\tau$ had the same histories of subordinates and original histories $o^{i,\circ}_\tau$ and $o^{i,\textcolor{gray}{\bullet}}_\tau$ were expanded using the same action and observation. Next, we need to show that the nested-belief states associated with the expanded histories $\langle o^{i,\circ}_\tau, u^i_\tau, z^i_{\tau+1}\rangle$ and $\langle o^{i,\textcolor{gray}{\bullet}}_\tau, u^i_\tau, z^i_{\tau+1}\rangle$  are also equivalent. The proof of this statement follows directly from the fact that the transition function from one stage to the next one depends on $\tilde{T}$, $T^{\cdot}$ and $\Omega^{\cdot}$. A careful inspection of these functions reveals that they depend on nested-occupancy states at player $i$ only through nested-belief states at player $i$, joint histories of superiors of player $i$, and actions of subordinates as demonstrated in Lemmas \ref{lem:markov} and \ref{lem:observation}. Histories of the current player are only used to select the action for that player. However, our histories of interest have the same optimal action set. So, assuming these histories take the same action does not hurt.  Consequently, if we let $s^{i,\circ}_\tau \sim_{\mathscr{B}_2} s^{i,\textcolor{gray}{\bullet}}_\tau$ then we know that $o^{i,\circ}_\tau \sim_{\mathscr{P}} o^{i,\textcolor{gray}{\bullet}}_\tau$, which ends the proof.
%
Which ends the proof for both propositions.
\end{proof}

\section{Proof of Theorem \ref{thm:error:bound}}
\label{appendix:thm:error:bound}

\begin{proof}
For simplicity, throughout the proof, we assume with no loss of generality that the central planner does not rely on public observations, so transition function $\pmb{T}$ is deterministic. Let $a^*_{0:}$ be an optimal joint policy with value functions $\upsilon^*_{0:}$. Let $s^*_\tau$ be the occupancy state generated under joint policy $a^*_{0:}$, with boundary condition $s^*_0 \doteq s_0$. Let $ \upsilon_{0:}$ be the value function that the PBVI algorithm produced over occupancy subsets $\tilde{S}_{0:}$. Then, it follows that:
\begin{align*}
\upsilon^*_0(s^*_0) -  \upsilon_0(s^*_0) &= \textstyle (\sum_{\tau=0}^{\ell-1} \gamma^\tau\cdot \pmb{R}(s^*_\tau,a^*_\tau)) -  \upsilon_0(s^*_0),\quad\text{(definition of $\upsilon^*_0(s^*_0)$)}\\
& = \textstyle (\sum_{\tau=0}^{\ell-1} \gamma^\tau\cdot \pmb{R}(s^*_\tau,a^*_\tau)) - \sum_{\tau=0}^{\ell-1} \gamma^\tau \cdot (\upsilon_\tau(s^*_\tau)-\upsilon_\tau(s^*_\tau)) -  \upsilon_0(s^*_0),\quad\text{(adding zero)}.
\end{align*} 
Next, we use the fact that $\upsilon_\ell(\cdot) \doteq 0$ to re-arrange terms:
\begin{align*}
&=\textstyle \sum_{\tau=0}^{\ell-1} \gamma^\tau \cdot\pmb{R}(s^*_\tau,a^*_\tau) + \left( \gamma^\ell\cdot \upsilon_\ell(s^*_\ell) + \sum_{\tau=1}^{\ell-1} \gamma^\tau \cdot \upsilon_\tau(s^*_\tau)\right) - \left(\gamma^0\cdot\upsilon_0(s^*_0)+\sum_{\tau=1}^{\ell-1}\gamma^\tau \cdot \upsilon_\tau(s^*_\tau)\right),\\
&=\textstyle \sum_{\tau=0}^{\ell-1} \gamma^\tau\cdot \pmb{R}(s^*_\tau,a^*_\tau) + \sum_{\tau=0}^{\ell-1} \gamma^{\tau+1} \cdot \upsilon_{\tau+1}(s^*_{\tau+1}) - \sum_{\tau=0}^{\ell-1}\gamma^\tau \cdot \upsilon_\tau(s^*_\tau),\\
&=\textstyle \sum_{\tau=0}^{\ell-1} \gamma^\tau\cdot  \left( \pmb{R}(s^*_\tau,a^*_\tau) +  \gamma \cdot \upsilon_{\tau+1}(s^*_{\tau+1}) -  \upsilon_\tau(s^*_\tau)\right).
\end{align*} 
Define $\upsilon^{\cdot}_\tau(s_\tau)\colon a_\tau \mapsto  \pmb{R}(s_\tau,a_\tau) +  \gamma \cdot \upsilon_{\tau+1}(\pmb{T}(s_\tau,a_\tau))$. It follows that 
\begin{align*}
\upsilon^*_0(s^*_0) -  \upsilon_0(s^*_0) &=\textstyle \sum_{\tau=0}^{\ell-1} \gamma^\tau\cdot  \left(  \upsilon^{a^*_\tau}_\tau(s^*_\tau) -  \upsilon_\tau(s^*_\tau)\right).
\end{align*} 
If we fix $s_\tau \doteq \argmin_{\tilde{s}_\tau\in \tilde{S}_\tau} \|s^*_\tau-\tilde{s}_\tau\|_1 $, then we know that $\|s^*_\tau-s_\tau\|_1 \leq \delta_{ \tilde{S}_{0:}}$ by definition of $\delta_{ \tilde{S}_{0:}}$. Using action-values $\upsilon^{a^*_\tau}_\tau(s_\tau)$ to add zero into the previous error bound results in:
\begin{align*}
\upsilon^*_0(s^*_0) -  \upsilon_0(s^*_0) &=\textstyle \sum_{\tau=0}^{\ell-1} \gamma^\tau\cdot  \left(   \upsilon^{a^*_\tau}_\tau(s^*_\tau)) -  \upsilon^{a^*_\tau}_\tau(s_\tau) +  \upsilon^{a^*_\tau}_\tau(s_\tau) -  \upsilon_\tau(s^*_\tau)\right).
\end{align*} 
Taking the best joint decision rule for $ \upsilon^{\cdot}_\tau(s_\tau)$ results in value $\upsilon_\tau(s_\tau)$ greater or equal to $ \upsilon^{a^*_\tau}_\tau(s_\tau)$, which leads to
\begin{align*}
\upsilon^*_0(s^*_0) -  \upsilon_0(s^*_0) &\leq\textstyle \sum_{\tau=0}^{\ell-1} \gamma^\tau\cdot  \left(   \upsilon^{a^*_\tau}_\tau(s^*_\tau) -  \upsilon^{a^*_\tau}_\tau(s_\tau) + \upsilon_\tau(s_\tau) -  \upsilon_\tau(s^*_\tau)\right).
\end{align*} 
Recall that under a fixed joint policy, value functions are linear functions of occupancy states, which allows us to re-arrange terms as follows:
\begin{align*}
\upsilon^*_0(s^*_0) -  \upsilon_0(s^*_0) &\leq\textstyle \sum_{\tau=0}^{\ell-1} \gamma^\tau\cdot  \left(   \upsilon^{a^*_\tau}_\tau(s^*_\tau)  -  \upsilon_\tau(s^*_\tau)+ \upsilon_\tau(s_\tau) -  \upsilon^{a^*_\tau}_\tau(s_\tau) \right)\\
&=\textstyle \sum_{\tau=0}^{\ell-1} \gamma^\tau\cdot  (   \upsilon^{a^*_\tau}_\tau  -  \upsilon_\tau ) \cdot  (  s^*_\tau  -  s_\tau ).
\end{align*} 
The application of the H\"{o}lder inegality,  the use of the definition of $\delta_{ \tilde{S}_{0:}}$, and the use of the bounded reward function $r(\cdot,\cdot)$, permit us to conclude:
\begin{align*}
\upsilon^*_0(s^*_0) -  \upsilon_0(s^*_0) &\leq\textstyle \sum_{\tau=0}^{\ell-1} \gamma^\tau\cdot  \|  \upsilon^{a^*_\tau}_\tau  -  \upsilon_\tau \|_\infty \cdot  \|  s^*_\tau  -  s_\tau \|_1\\
&=\textstyle \delta_{ \tilde{S}_{0:}}\sum_{\tau=0}^{\ell-1} \gamma^\tau\cdot  \|  \upsilon^{a^*_\tau}_\tau  -  \upsilon_\tau \|_\infty\\
&\leq\textstyle2c \delta_{ \tilde{S}_{0:}}\sum_{\tau=0}^{\ell-1} \gamma^\tau\sum_{t=\tau}^{\ell-1} \gamma^{t-\tau}\\
&=\textstyle 2c \delta_{ \tilde{S}_{0:}}\sum_{\tau=0}^{\ell-1} \sum_{t=\tau}^{\ell-1} \gamma^t\\
&=\textstyle 2c \delta_{ \tilde{S}_{0:}}\sum_{\tau=0}^{\ell-1} \frac{\gamma^\tau-\gamma^\ell}{1-\gamma}\\
&=\textstyle 2c \delta_{ \tilde{S}_{0:}}\frac{1}{1-\gamma}\sum_{\tau=0}^{\ell-1} (\gamma^\tau-\gamma^\ell)\\
&=\textstyle 2c \delta_{ \tilde{S}_{0:}}\frac{1}{1-\gamma}\sum_{\tau=0}^{\ell-1} \frac{1+\ell \gamma^{\ell+1}-(\ell+1) \gamma^\ell }{1-\gamma}\\
&=\textstyle 2c \delta_{ \tilde{S}_{0:}} \frac{1+\ell \gamma^{\ell+1}-(\ell+1) \gamma^\ell }{(1-\gamma)^2}.
\end{align*} 
Which ends the proof.
\end{proof}

\section{Multi-Player Benchmarks}
\label{sec:multi:player:benchmarks}

\paragraph{Multi-player Tiger.}
\new{
The $1$-player tiger problem was first introduced by \citet{kaelbling1998planning} and was later generalized to a $2$-player version by \citet{Nair-ijcai-03}. This game describes a scenario where players face two closed doors, one of which conceals a treasure while the other hides a dangerous tiger. Neither player knows which door leads to the treasure and which one to the tiger, but they can receive partial and noisy information about the tiger's location by listening. At any given time, each player can choose to open either the left or right door, which will either reveal the treasure or the tiger, and reset the game. To gain more information about the tiger's location, players can listen to hear the tiger on the left or right side, but with uncertain accuracy.  \\
We have extended this problem to an $n$-player version by incorporating hierarchical information-sharing and modifying the transition, observation, and reward models following \citet{Nair-ijcai-03}, while ensuring that the original $2$-player problem can still be recovered.
In this $n$-player version, only the reward function is not straightforwardly adapted.
Listening still costs $1$ per player, as in the original problem, while the penalty for opening the wrong door is now set to $-100 / n_w$ (with $n_w$ the number of players opening the bad door) and the reward for opening the good door is $10$ per player.
}

\paragraph{Multi-player Recycling Robot.}
The recycling robot task was first introduced by \citet{sutton2018reinforcement} as a single-player problem. Later on, \citet{amato2012optimizing} generalized it to a two-player version. The multi-player formulation requires robots to work together to recycle soda cans. In this problem, both robots have a battery level, which can be either high or low. They have to choose between collecting small or big cans and recharging their own battery level. Collecting small or big cans can decrease the robot's battery level, with a higher probability when collecting the big can. When a robot's battery is completely exhausted, it needs to be picked up and placed onto a recharging spot, which results in a negative reward. The coordination problem arises since robots cannot pick up a big can independently.
\old{
To solve this problem, an $n$-player Dec-POMDP was derived by allowing a big soda to be picked up only when all robots try to collect it simultaneously. The other transition and observation probabilities come from the $n$ independent single-player models introduced by \mbox{\citet{sutton2018reinforcement}.}
}
\new{
In our n-player version of the problem, picking up small cans still rewards $2$ per agent.
A reward of $5$ per agent is given if all agents synchronize to carry a big can, while a penalty of $10$ is given if some agents (but not all) try to carry a big can.
}

\paragraph{Multi-player Broadcast Channel.}
In 1996, \citet{ooi-CDC-96} introduced a scenario in which a unique channel is shared by $n$ players, who aim at transmitting packets. The time is discretized, and only one packet can be transmitted at each time step. If two or more players attempt to send a packet at the same time, the transmission fails due to a collision. In 2004, \citeauthor{HansenBZ04} extended this problem to a partially observable one, focusing on two players \cite{HansenBZ04}. We used similar adaptations to define a partially observable version of the original $n$-player broadcast channel.

\paragraph{Multi-player Grid3x3.}
This problem was first introduced by \citet{639686}. It involves two players who want to meet each other as soon as possible on a two-dimensional grid. Each player has five possible actions: moving north, south, west, east, or staying in place. To simulate an uncertain environment, each player's action has a fixed probability of being successful. Additionally, each player can only sense their own location and has no knowledge of the other player's location. To adapt the problem for multiple players, we placed M players on the grid, each with the same actions and perceptions as described above. The reward has been redefined as the largest number of players minus one present at one of the two meeting points. This way, the original problem can be retrieved for two players.

\section{Experiments}
\label{sec:experimental:results}

We conducted three sets of experiments to assess our findings:
\begin{enumerate}
\item To assess the exponential drop in time complexity of backups with respect to an increasing number of players, we maintain the average time required to perform a single backup, \cf Section \ref{subsec:abt:player} -- Average Backup Time for Increasing Players. 
\item To assess the exponential drop in time complexity of backups with respect to increasing horizons, we maintain the average time required to perform a single backup, \cf Section \ref{subsec:abt:horizon} -- Average Backup Time for Increasing Horizon. 
\item To assess the superiority of our findings with respect to the state-of-the-art approach to solve general decentralized partially Markov decision processes near-optimally, \cf Section \ref{subsect:sota} -- Against State-Of-The-Art Solvers.
\end{enumerate}

\subsection{Average Backup Time for Increasing Players}
\label{subsec:abt:player}

This section investigates the average computational time required to perform a single backup for increasing players, \cf Figures \ref{fig:backupTimeN:tiger},\ref{fig:backupTimeN:recycling},\ref{fig:backupTimeN:mabc}, and \ref{fig:backupTimeN:grid3x3}.  The experiments show that on all tested benchmarks,  hPBVI exhibits a reduction in time complexity compared to the other variants. Moreover, hPBVI can handle a larger number of agents (up to 9 for the Tiger, MABC, and Recycling) compared to the other variants, which are limited to a maximum of 5 agents. This time-complexity reduction in hPBVI is the result of our findings providing the ability to fully exploit the hierarchical information-sharing structure.  

\begin{figure}[!ht]
    \centering
	\label{fig:backupTimeN:tiger}
    \includegraphics[width=.8\textwidth]{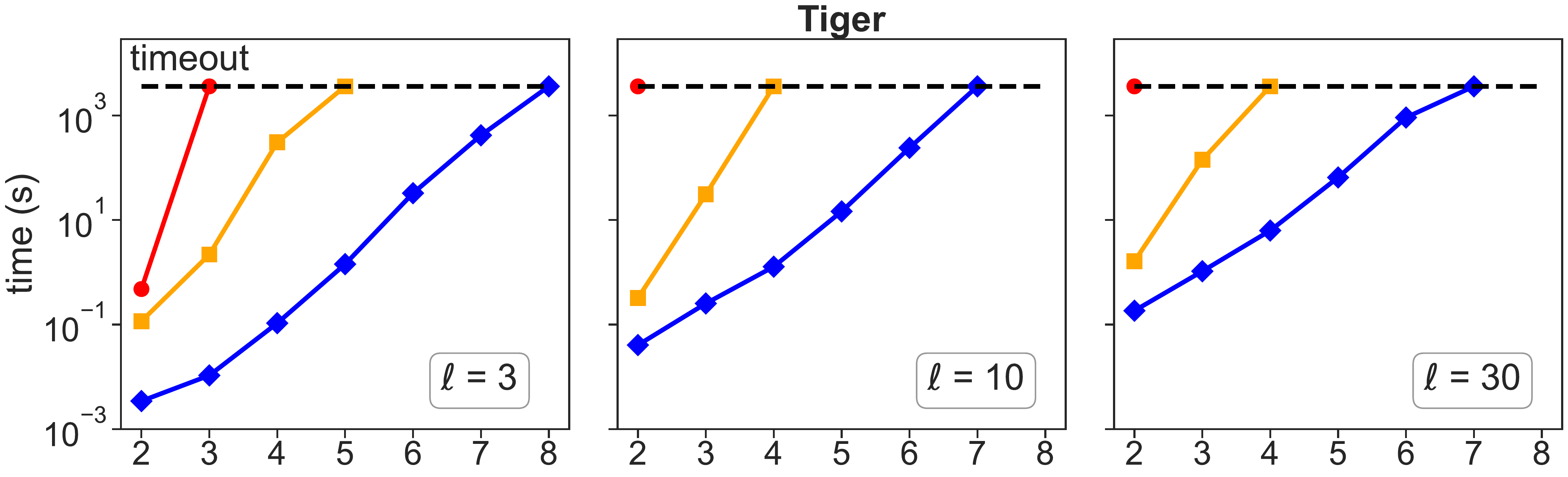}
    \caption{Average Backup Time for the tiger problem and different numbers of players.}
\end{figure}

\begin{figure}[!ht]	
    \centering
	\label{fig:backupTimeN:recycling}
    \includegraphics[width=.8\textwidth]{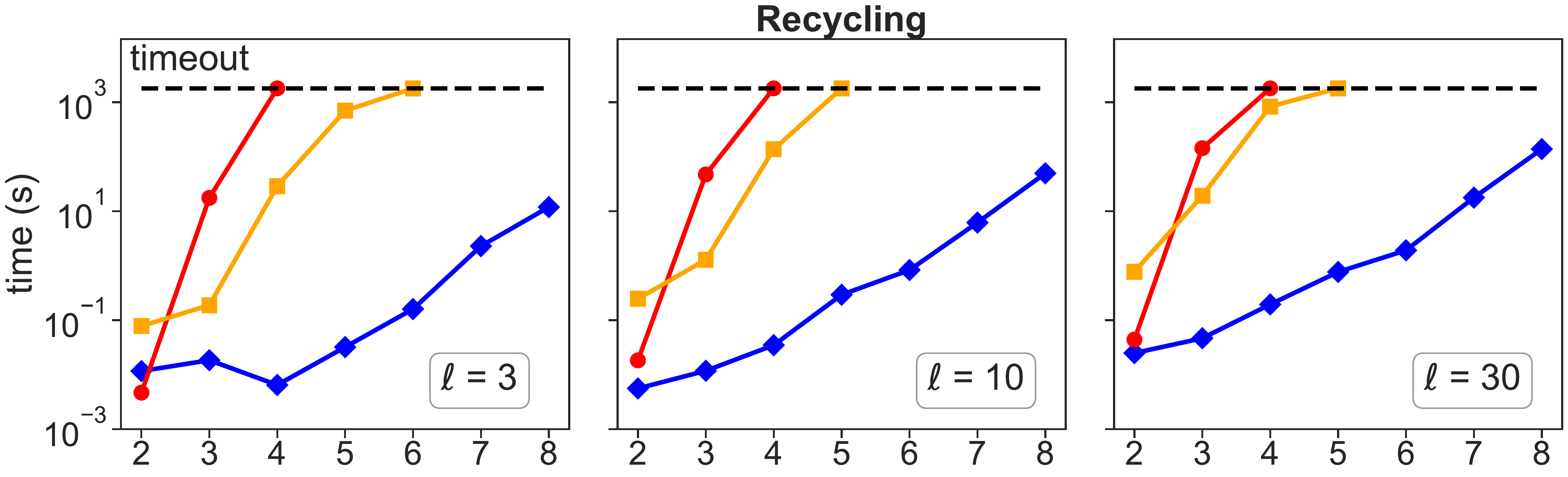}
    \caption{Average Backup Time for the recycling problem and different numbers of players.}
\end{figure}

\begin{figure}[!ht]	
    \centering
	\label{fig:backupTimeN:mabc}
    \includegraphics[width=.8\textwidth]{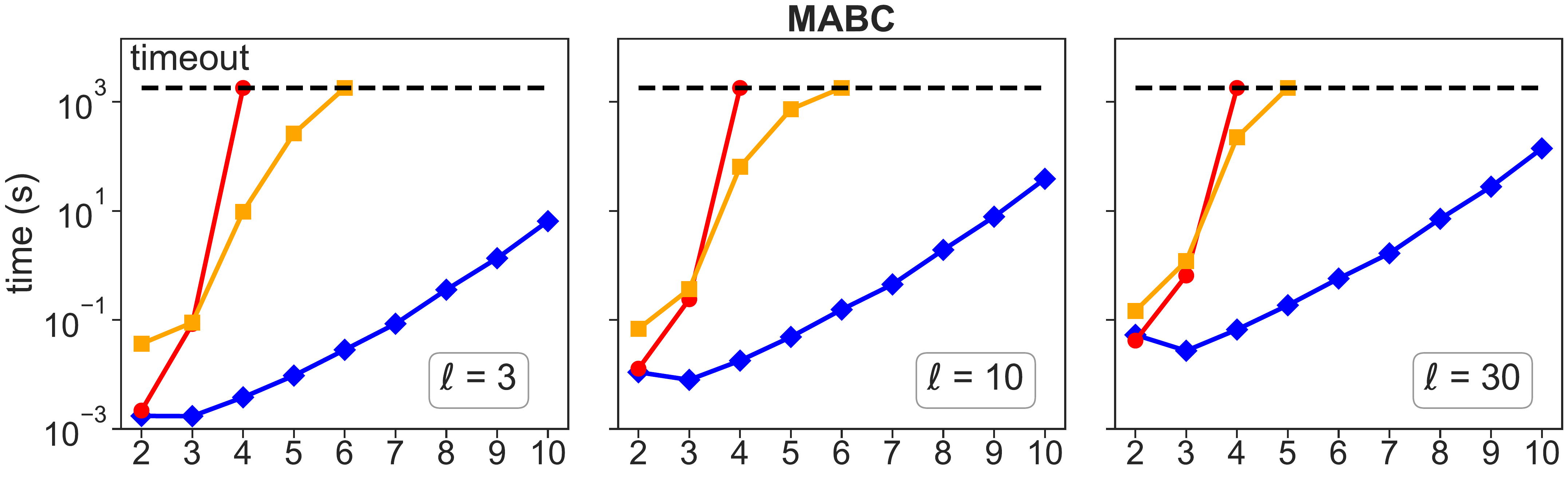}
    \caption{Average Backup Time for the mabc problem and different numbers of players.}
\end{figure}

\begin{figure}[!ht]
    \centering
	\label{fig:backupTimeN:grid3x3}
    \includegraphics[width=.8\textwidth]{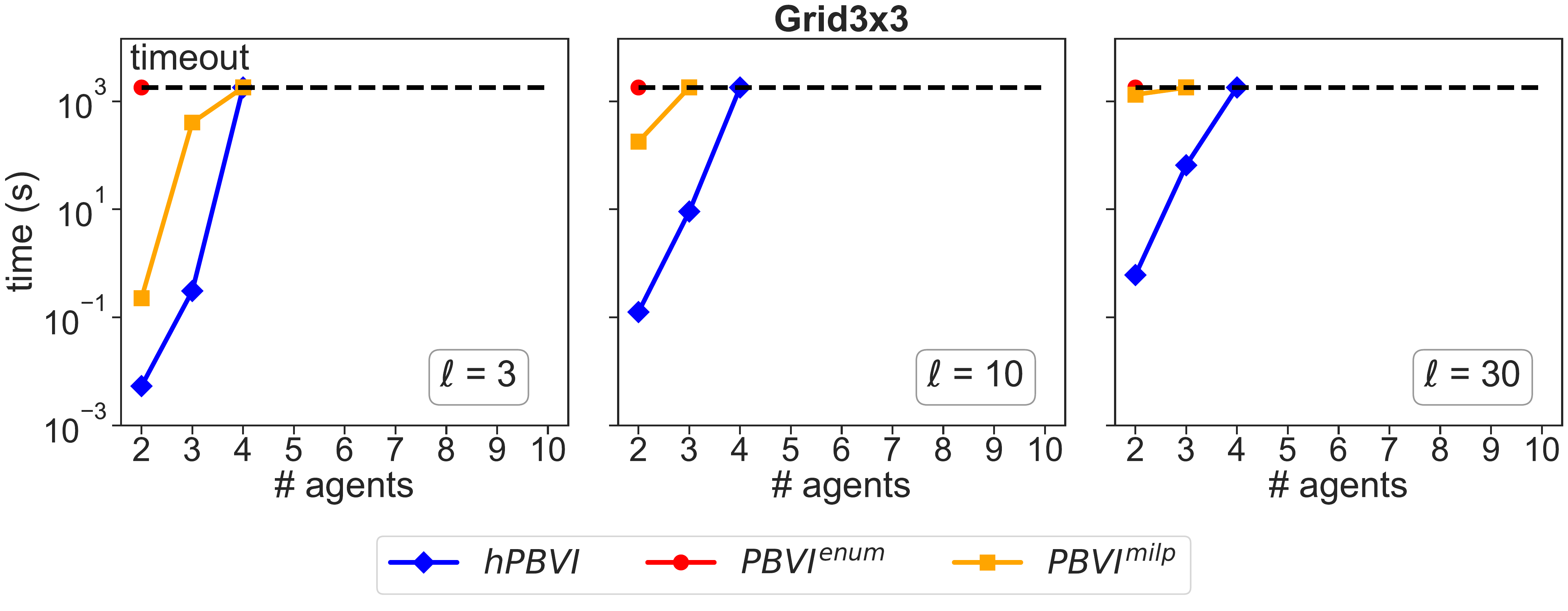}    
    \caption{Average Backup Time for the grid3x3 problem and different numbers of players.}
\end{figure}

\subsection{Average Backup Time for Increasing Horizons}
\label{subsec:abt:horizon}

This section investigates the average computational time required to perform a single backup for increasing horizons, \cf Figures \ref{fig:backupTimeH_tiger}, \ref{fig:backupTimeH_recycling}, \ref{fig:backupTimeH_mabc}, and \ref{fig:backupTimeH_grid3x3}. The experiments show once again that on all tested benchmarks,  hPBVI exhibits an exponential drop in time complexity compared to the other variants.  However, all three variants of the PBVI algorithm exhibit an increase in time complexity with respect to the planning horizon. This increase in time complexity is expected since, as time goes the number of backups also increases.

\begin{figure}[!ht]
    \includegraphics[width=1\textwidth]{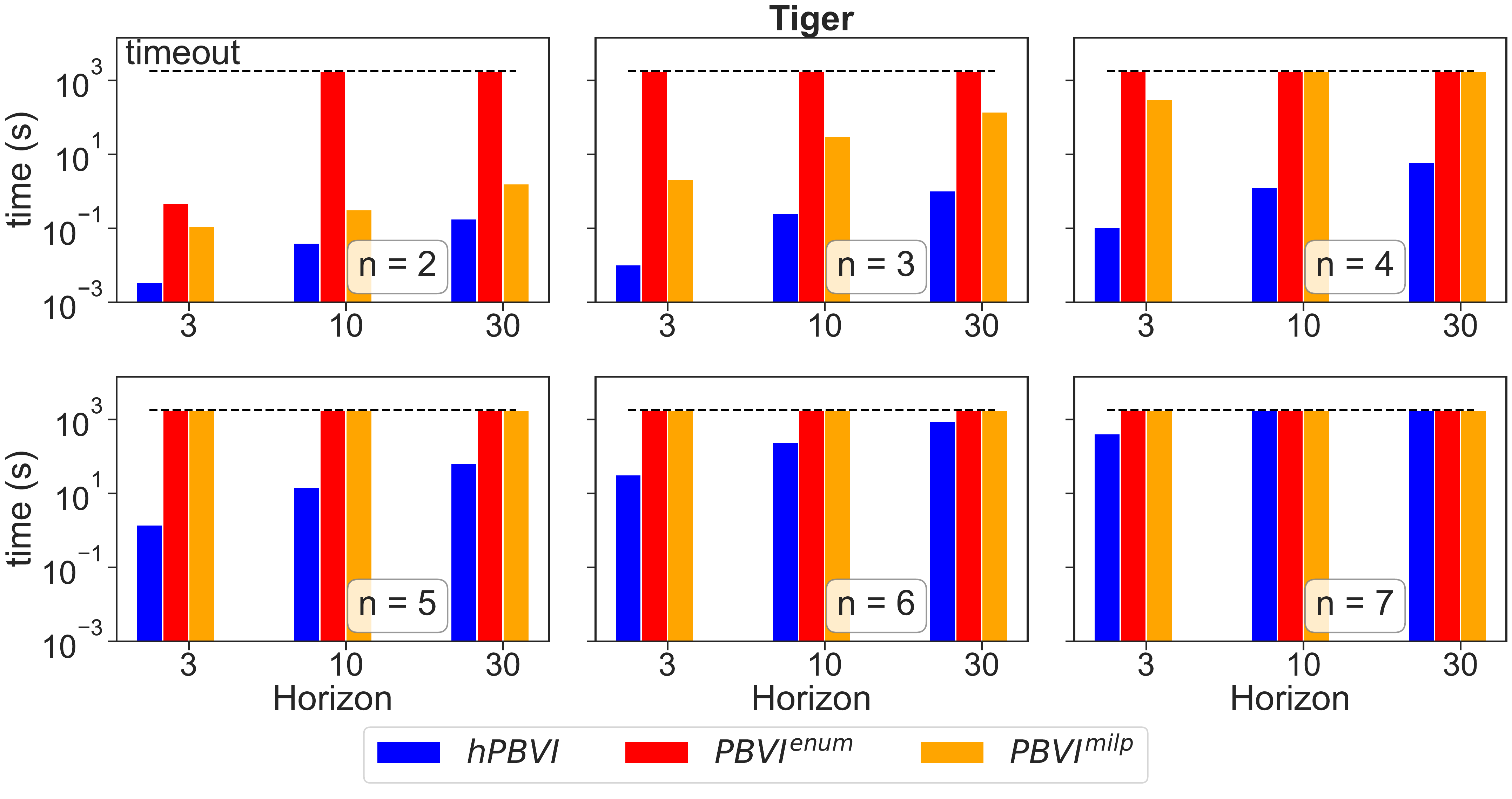}
    \label{fig:backupTimeH_tiger}
    \caption{Average backup time as a function of planning horizons for Tiger.}
\end{figure}

\begin{figure}[!ht]
    \includegraphics[width=1\textwidth]{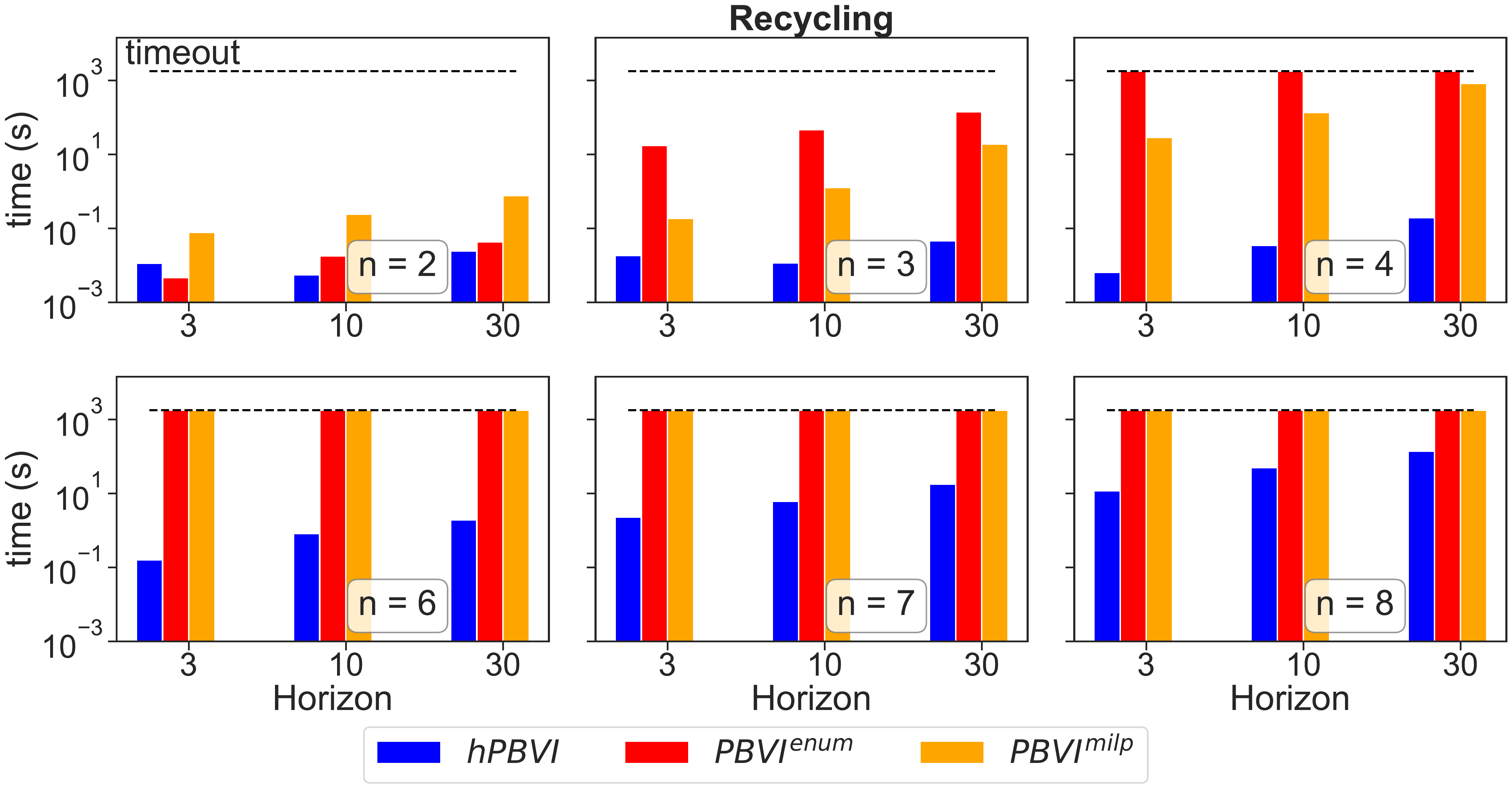}
    \label{fig:backupTimeH_recycling}
    \caption{Average backup time as a function of planning horizons for Recycling.}
\end{figure}

\begin{figure}[!ht]
    \includegraphics[width=1\textwidth]{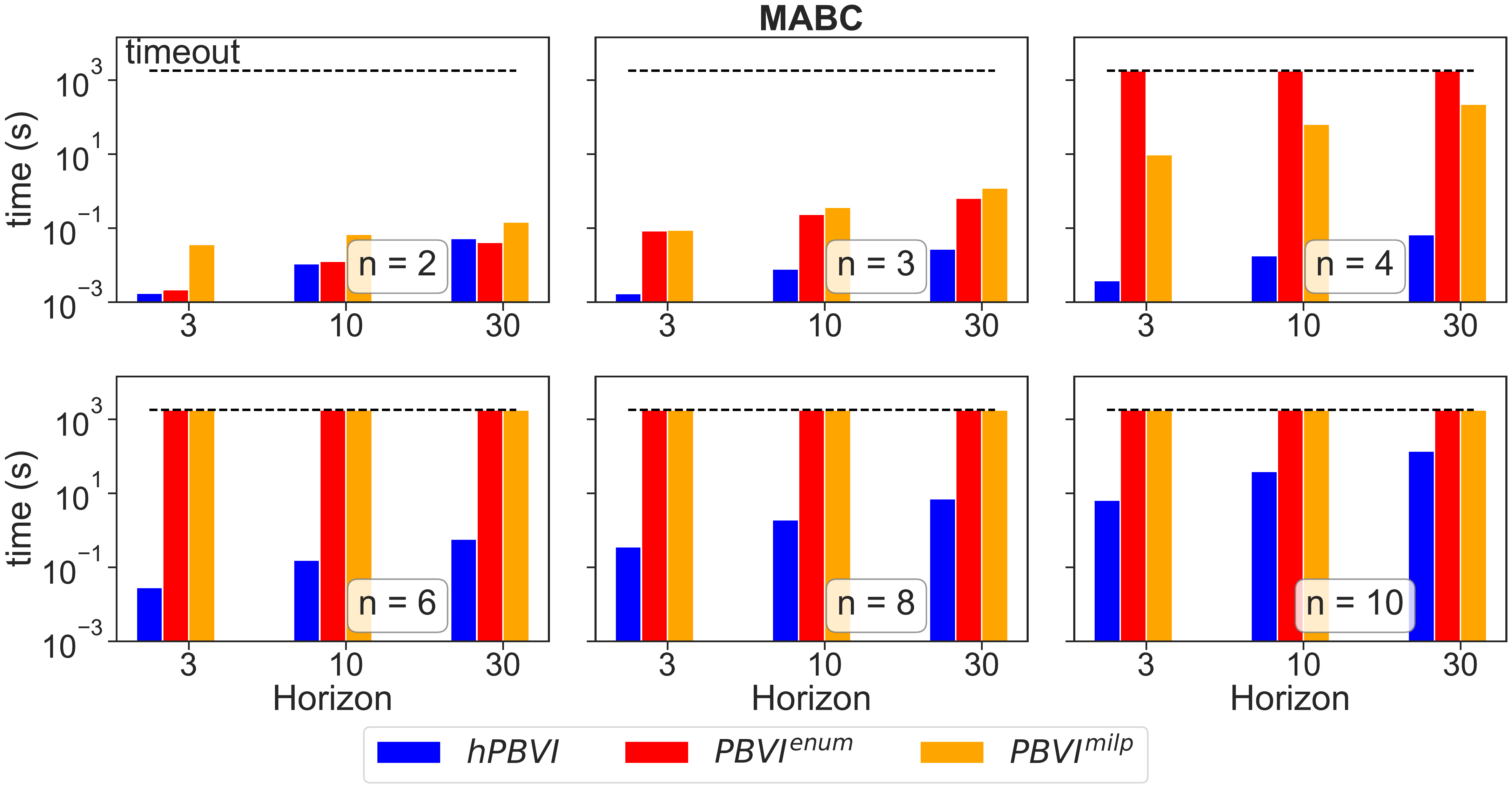}
    \caption{Average backup time as a function of planning horizons for MABC.}
    \label{fig:backupTimeH_mabc}
\end{figure}

\begin{figure}[!ht]
    \includegraphics[width=1\textwidth]{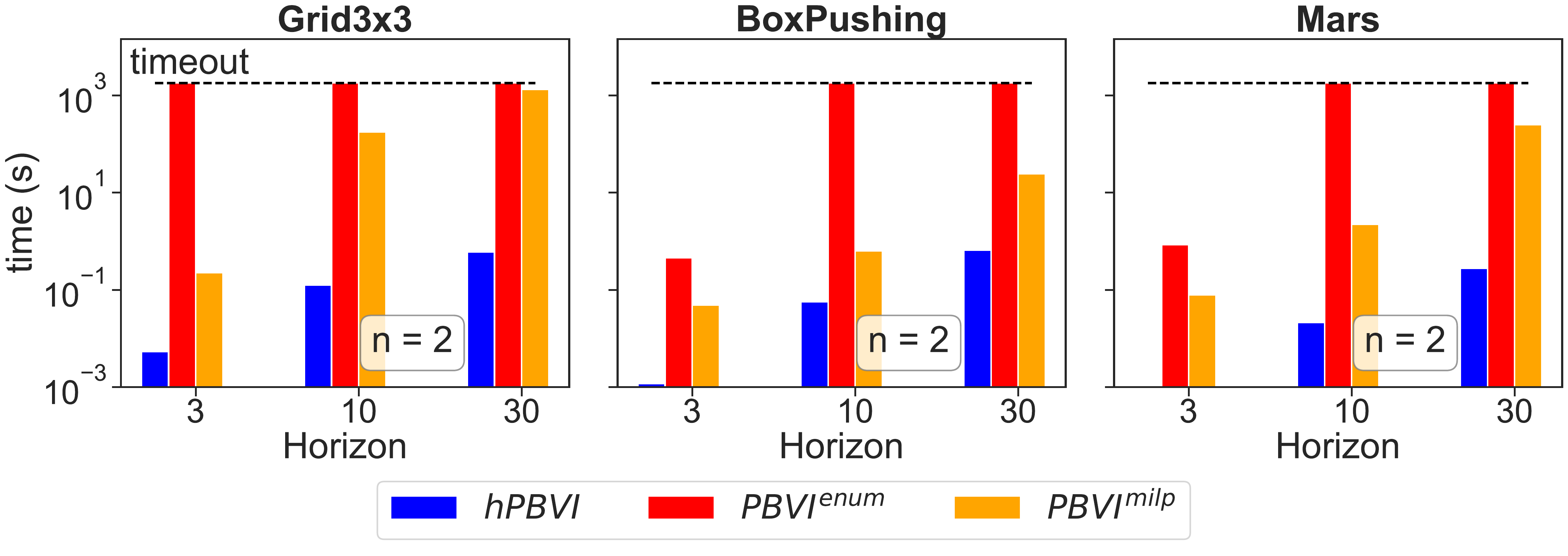}
    \caption{Average backup time as a function of planning horizons for Grid3x3, BoxPushing and Mars.}
    \label{fig:backupTimeH_grid3x3}
\end{figure}

\subsection{Against State-Of-The-Art Solvers}
\label{subsect:sota}

In this section, we compare our PBVI algorithm variants with two local algorithms, namely A2C and IQL, which are state-of-the-art and can handle a large number of players, as shown in Figures \ref{fig:anytimeCurves_tiger}, \ref{fig:anytimeCurves_recycling}, \ref{fig:anytimeCurves_mabc}, and \ref{fig:anytimeCurves_mabc}. However, these algorithms prioritize scalability over optimality and may get stuck in local optima. Our experiments demonstrate that hPBVI consistently outperforms all competitors in nearly all tested benchmarks in terms of convergence time and the value of the solution found within 30 minutes. In some weakly coupled domains, A2C and IQL find nearly optimal solutions close to those found by hPBVI. 

\begin{figure}[!ht]
	\includegraphics[width=1\textwidth]{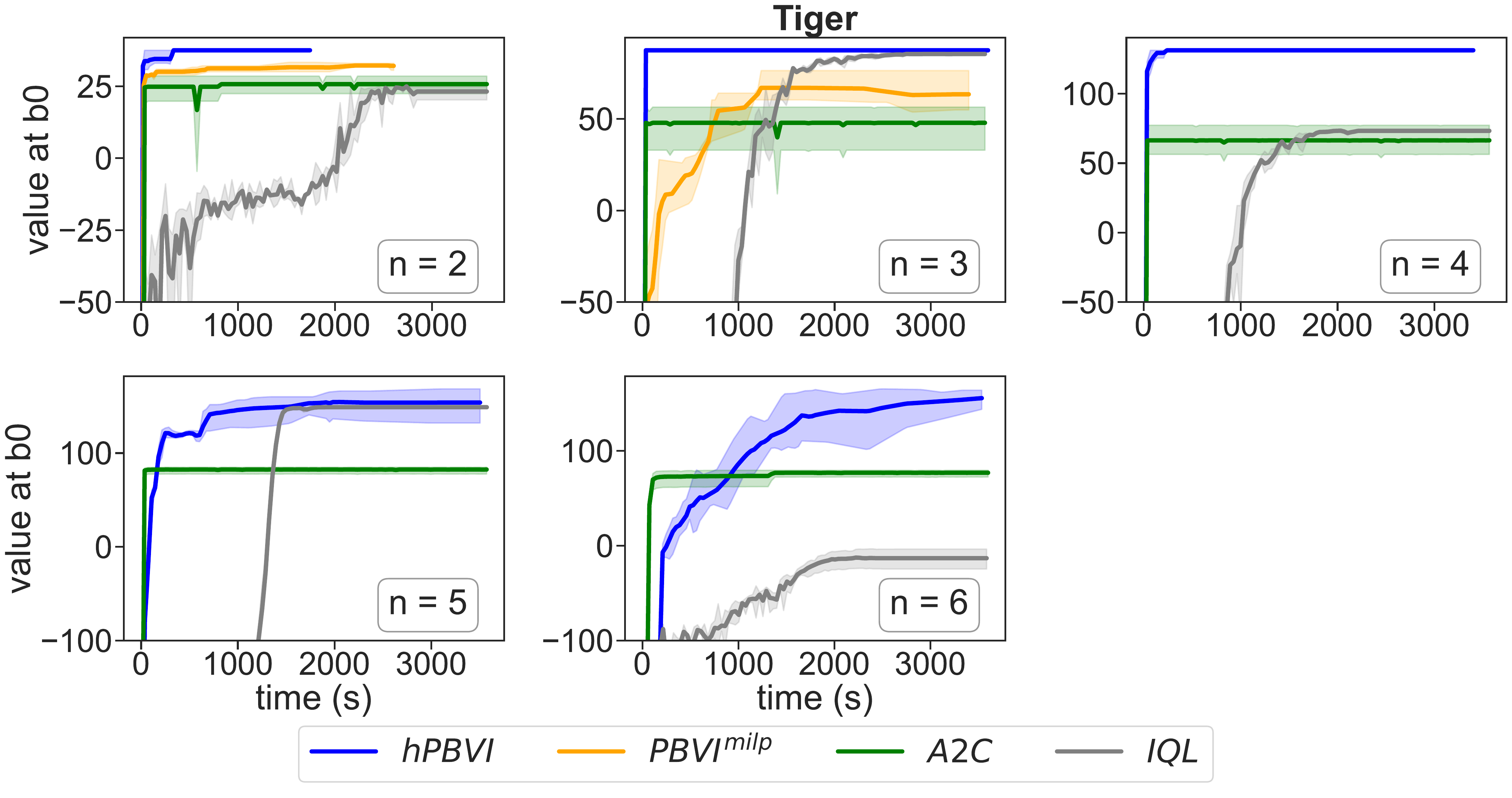}
	\caption{Anytime values for Tiger and $\ell = 30$.}
	\label{fig:anytimeCurves_tiger}
\end{figure}

\begin{figure}[!ht]
	\includegraphics[width=1\textwidth]{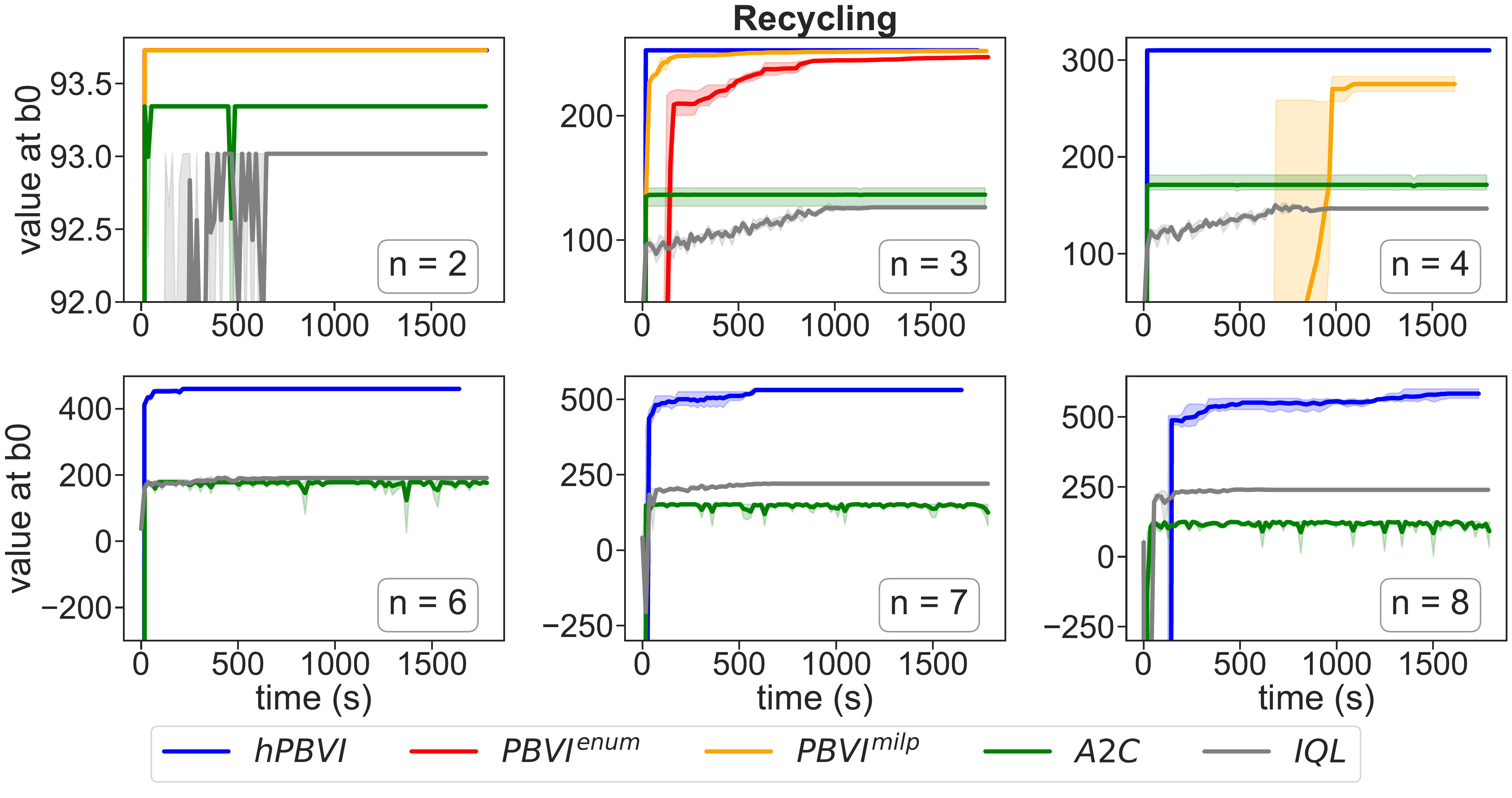}
	\caption{Anytime values for Recycling and $\ell = 30$.}
	\label{fig:anytimeCurves_recycling} 
\end{figure}

\begin{figure}[!ht]
	\includegraphics[width=1\textwidth]{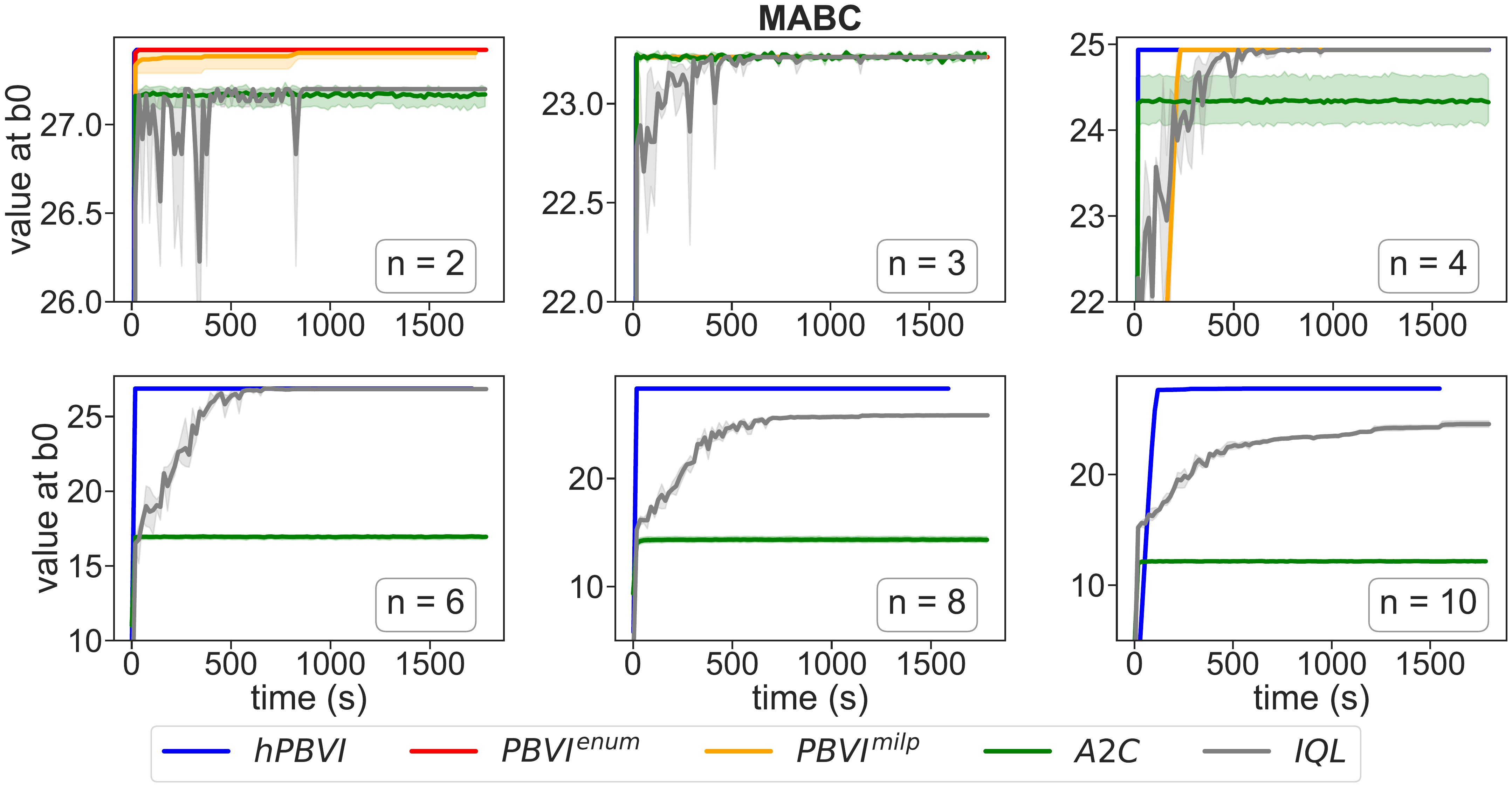}
	\caption{Anytime values for Multi-agent broadcast channel and $\ell = 30$.}
	\label{fig:anytimeCurves_mabc}
\end{figure}

\begin{figure}[!ht]
	\includegraphics[width=1\textwidth]{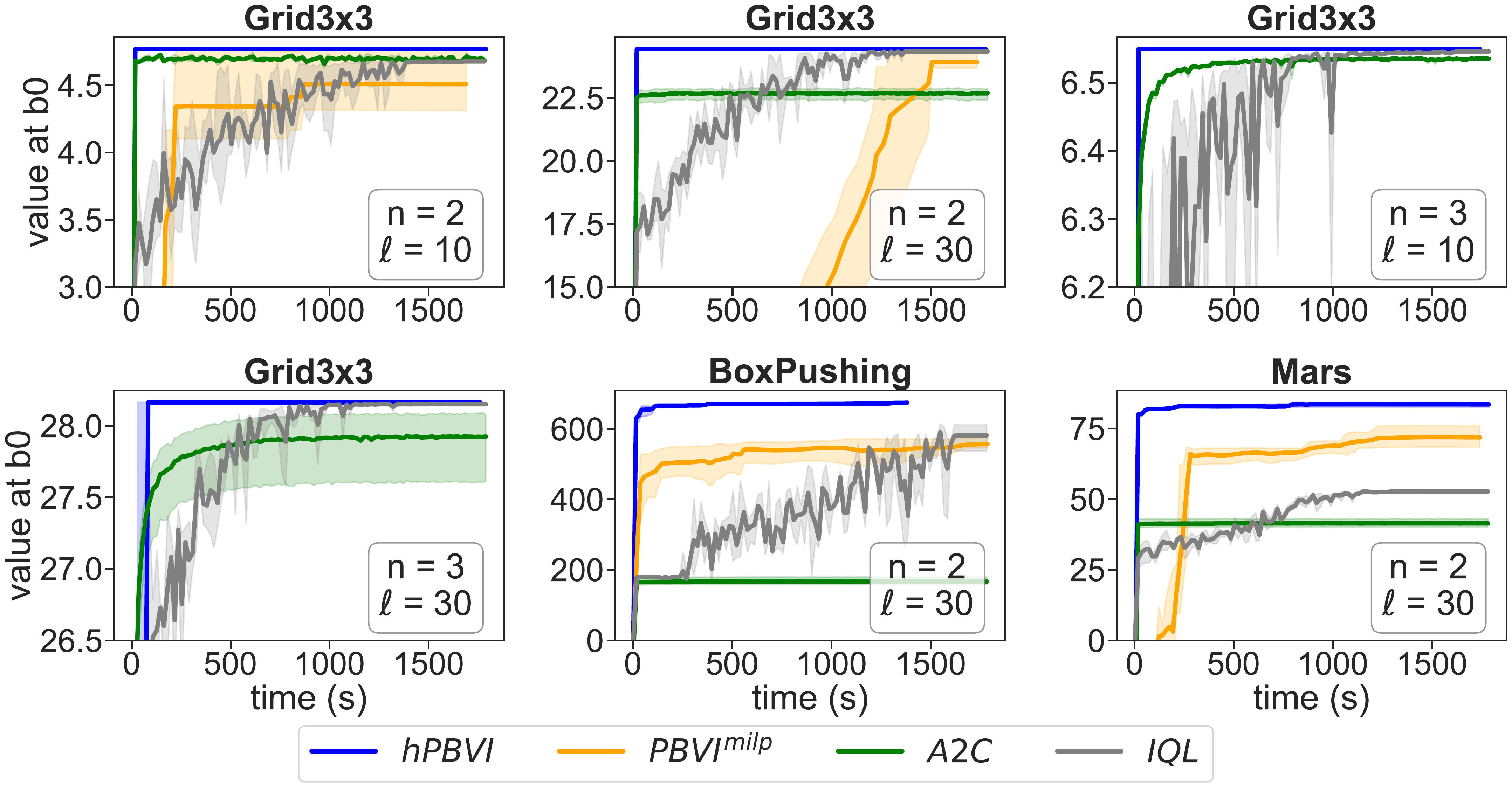}
	\caption{Anytime values for Grid3x3, BoxPushing and Mars.}
	\label{fig:anytimeCurves_grid3x3}
\end{figure}

%

\end{document}